\documentclass{article}
\usepackage[utf8]{inputenc}

\usepackage[framemethod=TikZ]{mdframed}
\usepackage{xcolor}
\usepackage{tikz}
\usetikzlibrary{shapes,decorations}
\usetikzlibrary {positioning}
\RequirePackage[OT1]{fontenc}
\RequirePackage{amsthm}
\RequirePackage[cmex10]{amsmath}
\usepackage{natbib}

\numberwithin{equation}{section}
\theoremstyle{plain}
\newtheorem{assumption}{Assumption}[section]
\newtheorem{lemma}{Lemma}[section]
\newtheorem{remark}{Remark}[section]

\usepackage{tabularx}
\usepackage{adjustbox}
\usepackage{booktabs}

\usepackage{array}
\newcolumntype{$}{>{\global\let\currentrowstyle\relax}}
\newcolumntype{^}{>{\currentrowstyle}}

\newcommand{\ci}{\perp\!\!\!\perp}

\usepackage{paralist}
\usepackage[section]{placeins}
\usepackage{amsfonts}
\usepackage{booktabs}
\usepackage{siunitx}
\usepackage{tabularx}
\usepackage{mathptmx}
\usepackage[11pt]{moresize}
\usepackage{algorithm}
\usepackage{amsbsy}
\usepackage{amsmath}
\usepackage{amsmath,amssymb,amsthm,bm}
\usepackage{amsthm}
\usepackage{cmap}
\usepackage{cases}
\usepackage{xr}
\usepackage{algorithm}

\usepackage[english]{babel}
\usepackage{enumitem}
\usepackage{epsfig}
\usepackage{eucal}
\usepackage[font=small]{caption}
\usepackage[T1]{fontenc}
\usepackage{graphicx}
\usepackage{lscape}
\usepackage[margin=1.35in]{geometry}
\usepackage{mathptmx}
\usepackage{mathtools}
\usepackage{natbib}
\usepackage{nccmath}
\usepackage[noend,]{algorithmic}
\usepackage{setspace}
\usepackage{sgame}
\usepackage{subcaption}
\usepackage{tikz}
\usepackage{url}
\usepackage{pdflscape}
\usepackage{afterpage}
\usepackage{quoting, lipsum}

\newcommand{\E}{\mathbb{E}}

\newcount\Comments
\theoremstyle{plain}

\newtheorem{theorem}{Theorem}[]

\newtheorem{corollary}[theorem]{Corollary}

\newtheorem{proposition}[theorem]{Proposition}
\newcounter{example}[section]
\newenvironment{example}[1][]{\refstepcounter{example}\par\medskip
   \noindent \textbf{Example~\theexample. #1} \rmfamily}{\medskip}

\newtheorem{definition}[]{Definition}

\theoremstyle{remark}
\renewcommand{\theremark}{\arabic{section}.\arabic{remark}}
\renewcommand{\theexample}{\arabic{section}.\arabic{example}}

\renewcommand{\thecorollary}{\arabic{section}.\arabic{corollary}}

\newcommand{\ba}{\begin{array}}
\newcommand{\ea}{\end{array}}

\makeatletter
\DeclareRobustCommand{\varamalg}{%
  \mathbin{\mathpalette\var@malg\perp}%
}

\newcommand\var@malg[2]{%
  \rlap{$\m@th#1#2$}\mkern6mu{#1#2}%
}
\makeatother

\newcommand{\bs}{\begin{align}\begin{split}\nonumber}
\newcommand{\bsnumber}{\begin{align}\begin{split}}
\newcommand{\es}{\end{split}\end{align}}

\newcommand{\Zeta}{\mathrm{Z}}

\newcolumntype{Y}{>{\centering\arraybackslash}X}
\usepackage[T1]{fontenc}
\usepackage{etoolbox}

\makeatletter
\patchcmd{\maketitle}
 {\def\@makefnmark}
 {\def\@makefnmark{}\def\useless@macro}
 {}{}
\makeatother
\usepackage{amsmath}
\makeatletter
\renewcommand*\env@matrix[1][*\c@MaxMatrixCols c]{%
  \hskip -\arraycolsep
  \let\@ifnextchar\new@ifnextchar
  \array{#1}}
\makeatother

\usepackage{pdflscape}
\usepackage{afterpage}
\usepackage{capt-of}

\newcommand{\EN}{{\mathbb{E}_{N}}}
\newcommand{\GN}{{\mathbb{G}_{N}}}

\newenvironment{continuance}[1]
  {\par\bigskip\noindent\textbf{Example #1 (continued)}\itshape}
  {\par}

\begin{document}
\linespread{1.5}

\title {Debiased Machine Learning of Aggregated Intersection Bounds and Other Causal Parameters}

\author{
	Vira Semenova\thanks{Email: vsemenova@berkeley.edu. First version: March 2023, arXiv ID  2303.00982, Vira Semenova ``Adaptive Estimation of Intersection Bounds: a Classification Approach''. For helpful discussions, the author is grateful to Isaiah Andrews, David Bruns-Smith, Yahu Cong, Denis Chetverikov, Federico Echenique, Bryan Graham, Michael Jansson, Hiroaki Kaido,  Désiré Kédagni, 
	 Toru Kitagawa, Patrick Kline, Soonwoo Kwon, Ying-Ying Lee, Lihua Lei, Demian Pouzo, Jim Powell,  Ashesh Rambachan, Jonathan Roth, Chris Shannon, Rahul Singh, Sophie Sun, Davide Viviano, Christopher Walters,	Mingduo Zhang and numerous seminar participants.  }}  

\maketitle

\begin{abstract}
This paper proposes a novel framework of aggregated intersection of regression functions, where the target parameter is obtained by averaging the minimum (or maximum)   of a collection of regression functions over the covariate space. Examples of such quantities include the lower and upper bounds on distributional effects (Fr\'echet-Hoeffding, Makarov) as well as the optimal welfare in statistical treatment choice problem \citep{QianMurphy}. The proposed estimator -- the envelope score estimator -- is shown to have an oracle property, where the oracle knows the identity of the minimizer  for each covariate value. I apply this result to the bounds in Roy model and Horowitz-Manski-Lee bounds with discrete outcome. The proposed approach performs well empirically on  the data from Oregon Health Insurance Experiment \citep{finkelstein}.
  \end{abstract}

\section{Introduction}

Economists are often interested in bounds on parameters when parameters themselves are not point-identified  \citep{Manski89,Manski90,Manski}. Examples include quantiles of heterogeneous treatment effects and other distributional measures beyond the mean \citep{FanPark}. Baseline or pre-treatment covariates often contain valuable information that can tighten these bounds \citep{ManskiPepper}. However, in practice, sharp bounds are rarely utilized because their estimators usually have non-standard distributions driven by noisy first-stage estimators of unknown conditional distributions. These challenges are not unique to partial identification and also arise in related areas, such as statistical treatment choice \citep{LuedtkeLaan, KitagawaTetenov, AtheyWager2, MbakopTabord}. 

This paper develops estimation and inference methods for the quantities taking the form
 \begin{align}
\label{eq:mainpsiintro}
\psi_0:= \E_X [\min_{t \in \mathcal{T}} \phi (t, \nu_0(X))],
\end{align}
where $X$ is a covariate vector,  $\mathcal{T}$ is a finite index set, and  $x \mapsto \nu_0(\cdot) = (\nu_{j0}(\cdot))_{j=1}^d$ is a $d$-dimensional nuisance parameter whose elements $\nu_{j0}(x)$ are functions of covariates, such as conditional expectations.  As the simplest case, one can think of the optimal welfare which appears in e.g., \cite{LuedtkeLaan} or sharp bound on distributional effects \citep{FanPark}. The paper's contribution is to deliver a debiased inference on $\psi_0$ that is first-order insensitive to the misclassification mistake in the identity of the binding constraint. In particular, its distribution is the same as if the true value of the minimizer \eqref{eq:mainpsiintro} were known. Additionally, the paper establishes the validity of a weighted bootstrap method, which holds the estimated minimizer fixed while bootstrapping the second-stage statistic, providing a valid distributional approximation.

The paper illustrates the usefulness of the proposed approach by considering two applications  in applied microeconomics. In particular, we discuss in detail a sharp version of Roy model bounds as studied in \cite{MourifieHenry} as well as Horowitz-Manski-Lee bounds with discrete-valued outcomes a version of which have been also studied in concurrent, independent work of \cite{kroft2024leeboundsmultilayeredsample}. Revisiting  Oregon Health Insurance Experiment \cite{finkelstein}, we find our methodology useful in determining the direction of treatment effect in the presence of non-response bias as well as tightening the bounds, echoing earlier work in \cite{SemSupp2}.

The rest of the paper is organized as follows. Section \ref{sec:litreview} gives a literature review. Section \ref{sec:setup} introduces the framework and provides two stylized examples. Section \ref{sec:overview} offers an informal preview of the results. Section \ref{sec:mainr} presents the formal asymptotic theory and discusses low-level condition for the margin assumption in the context of a single-index model with continuous covariates. Section \ref{sec:appl} applies the proposed theory to sharp bounds in the Roy model and Horowitz-Manski-Lee bounds in selection problems. Section \ref{sec:empirical2} provides numerical evidence for the methods developed in the article. All proofs are in Appendix \ref{sec:proofs}.

\subsection{Literature Review}
\label{sec:litreview}

This paper is related to two lines of research: partial identification and statistical treatment choice.

\paragraph{Bounds, Convex Optimization, and Directionally Differentiable Functionals.}   Set identification is a vast area of research,  encompassing a wide variety of approaches: linear and quadratic programming, random set theory, support function, and moment inequalities \citep{Manski90, ManskiPepper, Manski:2002, HaileTamer, CHT, BM, Molinari2008, CilibertoTamer, LeeBound, Stoye, AndrewsShiECMA, BMM2, CCMS, CherNeweySantos, Gafarov, kallus2020localized, li2022discordant, henry2023role, acerenza2023marginal, ban2021nonparametric, bartalotti2021identifying, JLS, fava2024}, see e.g. \cite{Molinari:2018} or \cite{MolinariHandbook} for a  review. In the context 
of distributional effects  \citep{Makarov, Manski,HSC, FanPark,FanPark2, Tetenov, FanZhu, FirpoRidder}, the first discussion of estimation can be traced to \cite{FanPark}, where, on p.945 they sketch a plug-in estimation approach  without statistical guarantees. Targeting the envelope function $\inf_{t \in T} s(t,x)$, the work by \cite{CLR} proposes a plug-in approach based on the least squares series estimators, where large sample inference is based on the strong approximation of a sequence of series or kernel-based empirical processes. Switching the focus from the envelope function to its best linear predictor, \cite{CCMS} proposes a root-$N$ consistent and uniformly asymptotically Gaussian estimator of the target parameter, relying on the first-stage series estimators. Finally, recent work by \cite{LeeSungwon} focuses on bounds on conditional distributions of treatment effects.  That is, most inference work focuses on the envelope function, rather than its mean value, which makes the lack of differentiability of $x \mapsto \min (x,0)$ at the kink point $x=0$ a common concern (e.g., \citet{FangSantos}).  Finally, the paper contributes to a growing literature on machine learning for bounds and partially identified models \citep{kallus2019assessing, jeong2020robust, SemJoE} and sensitivity analysis \citep{DornGuo, DornGuoKallus, Bonvini_2021, Bonvini_2022}, see e.g., \cite{Kennedy_review} for the review.

\paragraph{Statistical treatment choice. }  Statistical treatment choice is a vast area of research, focusing on two distinct questions: learning the best policy \citep{QianMurphy, KitagawaTetenov, MbakopTabord, AtheyWager2, Sun} for a given criterion function and inference on the optimal  value  of the criterion itself \citep[e.g.,][]{LuedtkeLaan} in an unconstrained policy class.
In the first stream, recent work focused on various robustness aspects of existing criteria functions \citep[e.g.,][]{ishihara2021evidence, adjaho2023externally} or targeting welfare criteria that are partially identified \citep{Stoye, Pu_2021, Cui_2021, kitagawa2023treatment, yata2023optimal, dadamo2022orthogonal, christensen2023optimal, benmichael2023policy, olea2023decision, cui2024policy}. For example, such criteria may arise in  asymmetric loss functions \citep{babii2021binary}, partial welfare ordering \citep{han2021optimal,firpo2023loss} or distributional welfare based on quantile treatment effects \citep{cui2024policy}.

\section{Setup}
\label{sec:setup}

This paper studies the aggregated intersection of regression functions
\begin{align}
\label{eq:mainpsi}
\psi_0 := \E_X \big[\min_{t \in \mathcal{T}} \phi(t, \nu_0(X)) \big]
\end{align}
where $X$ is a covariate vector, $\mathcal{T}$ is a \textit{finite} index set, and 
$x \mapsto \nu_0(\cdot) = (\nu_{j0}(\cdot))_{j=1}^d$ is a $d$-dimensional nuisance parameter whose elements $\nu_{j0}(x)$ are functions of covariates, such as conditional expectations.  
For each element $t$ of the set $\mathcal{T}$, $\phi(t, \cdot): \mathrm{R}^d \rightarrow \mathrm{R}$ is a \textit{known} scalar function of the vector $\nu_0$, which could represent a projection onto the Euclidean axis. 
This function can be expressed as a conditional expectation
\begin{align}
\label{eq:stx}
\phi(t, \nu_0(x)) = \E \big[ \rho(W, t, \xi_0(X)) \mid X = x \big]
\end{align}
where $W$ is the data vector and $\rho(W, t, \xi_0)$ is an observed random variable that depends on a nuisance parameter $\xi_0$. 

Examples of $\psi_0$ include Fréchet-Hoeffding bounds, \citep{Makarov} bounds on distributional effects, and sharp versions of \citep*{BalkePearl1994, BalkePearl1997} bounds. In statistical treatment choice, examples include optimal welfare in an unconstrained policy class \citep{Manski2004,QianMurphy, LuedtkeLaan,KitagawaTetenov}. As a stylized example, this section revisits optimal welfare in statistical treatment choice and Makarov bounds on distributional effects \citep{FanPark}.

\begin{example}[Optimal Welfare]
\label{ex:welfare}
Let $D$ be a discrete-valued treatment variable taking values in a finite set $\mathcal{D}$. Let $Y(d)$ be a potential outcome, and let $Y = \sum_{d \in \mathcal{D}} Y(d) 1\{D = d\}$ be the observed outcome. The data vector is $W = (D, X, Y)$. 
Let $m(d, x) = \E \big[ Y \mid D = d, X = x \big]$ be the conditional expectation function. Under the unconfoundedness assumption
\begin{align}
\label{eq:ci}
(Y(d))_{d \in \mathcal{D}} \ci D \mid X,
\end{align}
the conditional means of potential outcomes are identified as
$$
\E \big[ Y(d) \mid X = x \big] = m(d, x).
$$
The negative attained welfare is
$$
\psi_0 = -\E \big[ \max_{d \in \mathcal{D}} m(d, X) \big] = \E \big[ \min_{d \in \mathcal{D}} -m(d, X) \big],
$$
which is a special case of \eqref{eq:mainpsi} with $\mathcal{T} = \mathcal{D}$ and
\begin{align*}
\nu_0(x) &= (m(d, x))_{d \in \mathcal{D}}, \qquad \phi(d, v) = -v_d, \quad d \in \mathcal{D}.
\end{align*}
The "unbiased" signal for $-m(d, x)$ is the \citep{Robins} orthogonal score
\begin{equation}
\label{eq:robins}
\rho(W, d, \xi_0) = -\frac{1\{D = d\}}{\mu_{d0}(X)} \big( Y - m(d, X) \big) - m(d, X),
\end{equation}
where the propensity score is
\begin{align*}
\mu_{d0}(x) = \Pr(D = d \mid X = x)
\end{align*}
and the nuisance parameter is
\begin{align*}
\xi_0(x) = \big(\nu_0(x), (\mu_{d0}(x))_{d \in \mathcal{D}}\big).
\end{align*}
When the treatment is binary, that is, $\mathcal{D} = \{0, 1\}$, the parameter $\psi_0$ reduces to
$$
\psi_0 = -\E \big[ \max(m(1, X), m(0, X)) \big]
$$
and $\mu_{10}(X) + \mu_{00}(X) = 1 \text{ a.s. }$
\end{example}

\begin{example}[Makarov Bounds on Distributional Effects]
\label{ex:makarov}
Consider the setup of Example \ref{ex:welfare} with $\mathcal{D} = \{0, 1\}$. Let $F_1(\cdot \mid x)$ and $F_0(\cdot \mid x)$ be the conditional Cumulative Distribution Functions (CDFs) of the potential outcomes $Y(1)$ and $Y(0)$, respectively, identified under \eqref{eq:ci}. Let $F_{Y(1) - Y(0)}(d)$ be the CDF of the treatment effect $Y(1) - Y(0)$, and let $F_{Y(1) - Y(0)}(d \mid x)$ be the conditional CDF. As shown in \citep{Makarov,FanPark,FirpoRidder}, the sharp bounds on the CDF of the treatment effect $F_{Y(1) - Y(0)}(d)$ are
\begin{align}
\label{eq:reg2}
\pi_L(d) &:= \E \big[\sup_{t \in \mathrm{R}} \max(F_1(t \mid X) - F_0(t - d \mid X)_{-}, 0)\big] \\
\pi_U(d) &:= \E \big[\inf_{t \in \mathrm{R}} \min(F_1(t \mid X) - F_0(t - d \mid X)_{-}, 0) + 1\big]
\end{align}
where $F(\cdot)$ and $F(\cdot)_{-}$ is the right-hand limit (i.e., regular CDF) and  the left-hand limit of the CDF, respectively. If the distributions $Y \mid D = 1, X = x$ and $Y \mid D = 0, X = x$ have finite support, their respective CDFs are step functions with finitely many jumps whose locations on $x$-axis are denoted by $\mathcal{T}_1$ and $\mathcal{T}_0$, respectively.

Consider the share of subjects negatively affected by the treatment. It is upper and lower bounded as 
$$
\pi_L(0) \leq \Pr (Y(1) - Y(0) \leq 0) \leq \pi_U(0).
$$
The upper bound $\pi_U(0)$ is a special case of \eqref{eq:mainpsi} with $\mathcal{T} = \{0\} \cup \mathcal{T}_1 \cup \mathcal{T}_0$, $\nu^1_0(x) = (F_1(t \mid X = x))_{t \in \mathcal{T}}$ and $\nu^0_0(x) = ({F_0(t \mid X = x)}_{-})_{t \in \mathcal{T}}$ and 
$$
\phi(t, v^1, v^0) = v^1_{t} - v^0_{t}, \quad t \in \{0\} \cup \mathcal{T}_1 \cup \mathcal{T}_0.
$$
The ``unbiased'' signal is the \citep{Robins}-type orthogonal score
\begin{align}
\label{eq:robins2}
\rho(W, t, \xi_0) &= \frac{D 1\{Y \leq   t\}}{\mu_{10}(X)} \big(1\{Y \leq  t\} - F_1(t \mid X)\big) + F_1(t \mid X) \\
&\quad - \frac{(1 - D) 1\{Y < t\}}{\mu_{00}(X)} \big(1\{Y < t\} - F_0(t \mid X)_{-}\big) - F_0(t \mid X)_{-}, \quad t \in \mathcal{T}_1 \cup \mathcal{T}_0
\end{align}
and $\rho(W, 0, \xi_0) = 0 \text{ a.s.}$. The propensity score is as in Example \ref{ex:welfare}, and the nuisance parameter is
$$
\xi_0(x) = (\nu^1_0(x),\nu^0_0(x), \mu_{10}(x), \mu_{00}(x))
$$
and $\mu_{10}(X) + \mu_{00}(X) = 1 \text{ a.s. }$.
\end{example}

Example \ref{ex:makarov} describes \cite{Makarov} bounds on the treatment effect CDF, previously studied in \cite{FanPark} and \cite{FirpoRidder}, among others. A special case of this example with binary outcomes was studied in \cite{kallus2022whats}, who proposed debiased inference for $\pi_L(0)$ and $\pi_U(0)$. Other interesting examples of this parameter are described in Section \ref{sec:appl}.

\section{Overview of Estimation and Inference}
\label{sec:overview}

In this section, I introduce the estimator of the parameter of interest and describe two inferential approaches. Let me briefly review the notation. Recall that $W$ is a data vector, $\mathcal{T}$ is an index set, and each function $\phi(t, \nu_0(x))$ is a conditional expectation function of an observed random variable $\rho(W, t, \xi_0)$
$$
\phi(t, \nu_0(x)) = \E \big[ \rho(W, t, \xi_0) \mid X = x \big].
$$
The identity of the minimizer is assumed unique
\begin{align}
\label{eq:unique}
t_0(x) := \arg \min_{t \in \mathcal{T}} \phi(t, \nu_0(x))
\end{align}
almost surely in $P_X$. 
The \textit{envelope regression function} is
$$
\min_{t \in \mathcal{T}} \phi(t, \nu_0(x)) = \phi(t_0(x), \nu_0(x)).
$$
Notice that this function can be written as
$$
\min_{t \in \mathcal{T}} \phi(t, \nu_0(x)) = \sum_{t \in \mathcal{T}} \phi(t, \nu_0(x)) 1\{t = \arg \min_{t \in \mathcal{T}} \phi(t, \nu_0(x))\}.
$$
Replacing each function $\phi(t, \nu_0(x))$ by its respective "unbiased signal" $\rho(W, t, \xi_0)$ gives the \textit{envelope moment} function
$$
 \sum_{t \in \mathcal{T}} \rho(W, t, \xi_0) 1\{t = \arg \min_{t \in \mathcal{T}} \phi(t, \nu_0(X))\} 
$$
where $\xi_0$ is the true value of the nuisance parameter $\xi$. By the law of iterated expectations,
$$
\psi_0 = \E[\rho(W, t_0, \xi_0)] = \E_X \big[\E[\rho(W, t_0, \xi_0) \mid X]\big] = \E_X \big[\min_{t \in \mathcal{T}} \phi(t, \nu_0(X))\big].
$$
The paper relies on standard cross-fitting \citep{schick1986asymptotically}, as commonly used in debiased machine learning \citep{chernozhukov2016double, AtheyWager2, LRSP}.

\begin{definition}[Cross-Fitting]
\mbox{}
\label{sampling}
\begin{compactenum} 
\item For a random sample of size $N$, denote a $K$-fold random partition of the sample indices $[N] = \{1, 2, \dots, N\}$ by $(J_k)_{k=1}^K$, where $K$ is the number of partitions and the sample size of each fold is $n = N / K$. For each $k \in [K] = \{1, 2, \dots, K\}$, define $J_k^c = [N] \setminus J_k$.
\item For each $k \in [K]$, construct an estimator $\widehat{\xi}_k = \widehat{\xi}(W_{i \in J_k^c})$ of the nuisance parameter $\xi_0$ using only the data $\{W_i : i \in J_k^c\}$. Define the first-stage fitted values
\begin{align}
\label{eq:tkxi} 
\widehat{t}_i = \widehat{t}(X_i) = \widehat{t}_k(X_i) := \arg \min_{t \in \mathcal{T}} \phi(t, \widehat{\nu}_k(X_i)), \quad i \in J_k, \\
\widehat{\xi}_i := \widehat{\xi}(X_i) = \widehat{\xi}_k(X_i), \quad i \in J_k. \nonumber
\end{align}
\end{compactenum}
\end{definition}

\begin{definition}[Estimator]
\label{def:envscore}
Given the first-stage fitted values, define
$$
\widehat{\psi} := \frac{1}{N} \sum_{i=1}^N \rho(W_i, \widehat{t}_i, \widehat{\xi}_i).
$$
\end{definition}

\begin{definition}[Multiplier Bootstrap]
\label{def:bb}
Let $(e_i)_{i=1}^N$ be a sequence of i.i.d. Exp(1) random variables independent of the data. Define
$$
\widetilde{\psi} = \frac{1}{N} \sum_{i=1}^N \frac{e_i}{\bar{e}} \rho(W_i, \widehat{t}_i, \widehat{\xi}_i), \label{eq:psi:boot}
$$
where $\bar{e} = N^{-1} \sum_{i=1}^N e_i$.
\end{definition}

Under some conditions on the nuisance parameter discussed below, the proposed estimator enjoys the following properties
\begin{enumerate}

\item With probability (w.p.) $\rightarrow 1$, the estimator converges at a root-$N$ rate:
$$
| \widehat{\psi} - \psi_0 | = O_P(1/\sqrt{N}) = o_P(1). \label{eq:urate}
$$

\item The estimator $\widehat{\psi}$ is asymptotically linear:
$$
\sqrt{N} (\widehat{\psi} - \psi_0) = \sqrt{N} \bigg( N^{-1} \sum_{i=1}^N \rho(W_i, t_0, \xi_0) - \psi_0 \bigg) + o_P(1),
$$
and, therefore, asymptotically Gaussian:
$$
\sqrt{N} (\widehat{\psi} - \psi_0) \Rightarrow^d N(0, V_0). \label{eq:limit}
$$
Its asymptotic variance:
\begin{equation}
V_0 := \E \big[\rho^2(W, t_0(X), \xi_0(X))\big] - \psi_0^2 \label{eq:psieta0}
\end{equation}
can be estimated by the sample analog:
\begin{equation}
\widehat{V} = N^{-1} \sum_{i=1}^N \rho^2(W_i, \widehat{t}_i, \widehat{\xi}_i) - \widehat{\psi}^2. \label{eq:varhat}
\end{equation}
\end{enumerate}

The paper establishes the theoretical framework for two inferential approaches. A plug-in $100(1-\alpha)\%$ confidence interval (CI) for $\psi_0$ can be constructed as
\begin{equation}
\label{eq:plugin}
CI_{1-\alpha} := \big(\widehat{\psi} - z_{1-\alpha/2} \sqrt{\widehat{V}/N}, \, \widehat{\psi} + z_{1-\alpha/2} \sqrt{\widehat{V}/N}\big),
\end{equation}
where $z_{1-\alpha}$ is the $(1-\alpha)$-quantile of $N(0, 1)$. As shown in Theorem \ref{thm:consistency}, the estimator $\widehat{V}$ is consistent for $V_0$, which implies
$$
\Pr \big(\psi_0 \in CI_{1-\alpha}\big) \rightarrow 1-\alpha, \quad N \rightarrow \infty.
$$
The plug-in variance estimator may be sensitive to biased estimation of $\xi$, which could affect the coverage of the plug-in confidence interval in small samples. An alternative to the plug-in procedure is to consider a bootstrap analog of the estimator 
$\widehat \psi$. A bootstrap confidence interval $CI^b_{1-\alpha}$ can be constructed as
\begin{equation}
\label{eq:boot}
CI^b_{1-\alpha} := \big(\widehat{\psi} + N^{-1/2} \widehat{C}_{\alpha/2}, \, \widehat{\psi} + N^{-1/2} \widehat{C}_{1-\alpha/2}\big),
\end{equation}
where the critical values $\widehat{C}_{\alpha/2}$ and $\widehat{C}_{1-\alpha/2}$ are the $\alpha/2$ and $1-\alpha/2$ quantiles of the bootstrapped statistic $\sqrt{N} (\widetilde{\psi} - \widehat{\psi})$. Thus
\begin{equation}
\Pr \big(\psi_0 \in CI^b_{1-\alpha}\big) \rightarrow 1-\alpha, \quad N \rightarrow \infty. \label{eq:coverageboot}
\end{equation}

\begin{remark}[Uniqueness of Minimizer]
\label{rm:unique}
The paper's results rely on the assumption that the minimizer $t_0(X)$ in \eqref{eq:unique} is almost surely unique. In the context of Example \ref{ex:welfare}, this condition simplifies to
\begin{align}
\label{eq:uniquemain}
P_X(m(1, X) - m(0, X) \neq 0) = 1.
\end{align}
This assumption is fundamental as it allows us to stay within the standard Gaussian framework. If this condition holds, the parameter $\psi_0$ is a pathwise differentiable parameter with a finite efficiency bound \citep{LuedtkeLaan} (cf. Lemma \ref{lem:LvdL} in Appendix). Otherwise, regular estimators of the optimal welfare may not exist \citep{HiranoPorter2012}.
\end{remark}

Requiring the minimizer to be unique is a non-trivial restriction on the data generating process. Remark \ref{rm:smooth} sketches a smoothing approach that could be used if condition \eqref{eq:uniquemain} is not plausible.

\begin{remark}[Smoothing Alternative]
\label{rm:smooth}
Following \citep{levis2023covariateassisted}, consider a log-sum-exp (LSE) function
\begin{equation}\label{eq:LSE}
g_{\widetilde{\kappa}}(\boldsymbol{v}) = \frac{1}{\widetilde{\kappa}}\log \left( \exp^{\widetilde{\kappa} v} + 1 \right), \text{ for } \boldsymbol{v} \in \mathbb{R},
\end{equation}
where $\widetilde{\kappa}$ is a tuning parameter and $\boldsymbol{v} \in \mathbb{R}$ is the argument. Noting that
$$
\max \left\{ \boldsymbol{v}, 0 \right\}<g_{\widetilde{\kappa}}(\boldsymbol{v}) \leq \max \left\{\boldsymbol{v}, 0 \right\}+\frac{\log 2}{\widetilde{\kappa}}
$$
allows to bound the approximation bias $\E [g_{\widetilde{\kappa}}(\nu_0(X)) ] - \E [ \max (m(1,X)- m(0,X),0)]$. The concurrent and independent  work by \citep{levis2023covariateassisted}  develops asymptotic efficiency theory for the smoothened analog.
\end{remark}

\section{Theoretical Results}
\label{sec:mainr}

Section \ref{sec:ass} states the assumptions required for asymptotic theory. Section \ref{sec:ma} describes a key condition required for asymptotic theory and verifies it for  single-index models.  Section \ref{sec:res} states the theoretical results.

\subsection{Assumptions}

\label{sec:ass}

Assumption \ref{ass:smallbias} ensures that the moment functions $\rho(W, t, \xi)$ are robust to first-order biases in the nuisance parameter $\xi$ uniformly over the index set $\mathcal{T}$. 

\begin{assumption}[Small Bias Condition]
\label{ass:smallbias}
There exists a sequence $\epsilon_N = o(1)$, such that with probability at least $1-\epsilon_N$, for every partition index $k \in [K]$, the first stage estimate $\widehat{\xi}_k$, obtained by cross-fitting, belongs to a shrinking neighborhood of $\xi_0$, denoted by $\Xi_N$. Uniformly over $\Xi_N$, the following mean square convergence holds:
\begin{align}
\label{eq:bn}
 B_N=\sup_{ \xi \in \Xi_N} \sup_{t \in \mathcal{T}} \sup_{x \in \mathcal{X}} \sqrt{N} |\E [ \rho(W,t,\xi)- \rho(W,t,\xi_0) \mid X=x] | = o(1). 
 \end{align}
 Furthermore, the second-order terms are bounded as
 \begin{align}
 \label{eq:lambdan}
 \Lambda_N =\sup_{ \xi \in \Xi_N} \sup_{t \in \mathcal{T}} \sup_{x \in \mathcal{X}}  \E [(\rho(W,t,\xi)- \rho(W,t,\xi_0))^2 \mid X=x] = o(1). 
 \end{align}
\end{assumption}

The small bias property is often attained via orthogonalization, a technique that has been widely studied in, e.g. \cite{Newey1994}, \cite{chernozhukov2016double}, and \cite{LRSP}. For example, if   $\rho(W, t, \xi) = \rho(W,t)$ does not involve any nuisance parameters,  Assumption \ref{ass:smallbias} is automatically satisfied.

\begin{continuance}{\ref{ex:welfare}}
Let $\mathcal{D} =\{0,1\}$. Consider an IPW-type signal of \citep{HIR2003}   taking the form 
$$ \rho(W,1) =  \dfrac{D}{\mu_{10}(X)} Y, \quad  \rho(W,0) = \dfrac{1-D}{\mu_{00}(X)} Y$$ 
where the propensity score is assumed known. Thus, Assumption \ref{ass:smallbias} is automatically satisfied with $B_N = \Lambda_N=0$.
\end{continuance}

 Assumption \ref{ass:smallbias} is shown to be satisfied if the signal is  orthogonal with respect to the nuisance function $\xi_0$ whose estimator converges at a $o(N^{-1/4})$ mean square rate (e.g., \citep{SemCher}).

\begin{continuance}{\ref{ex:welfare}}
Let $\mathcal{D} = \{0,1\}$.  Consider the Robins-type signal of \citep{Robins} with $\rho(W,1,\xi)$ and $\rho(W,0,\xi)$ defined as in \eqref{eq:robins}. Suppose each function $m(1,X)$ and $m(0,X)$ is estimated at a mean square rate $m_N$, and the function $\mu_{10}(X)$ is estimated at rate $\mu_N$. Under  Assumption 4.11 in \citep{SemCher}, Assumption \ref{ass:smallbias} holds if 
$$
 B_N =  O(\sqrt{N} m_N \cdot \mu_N)  =o(1), \qquad \Lambda_N = O (m_N +  \mu_N) = o(1).
$$ 
\end{continuance}

Assumption \ref{ass:rate} is a rate condition on the nuisance parameter $\nu$ entering \eqref{eq:mainpsi}. 

\begin{assumption}[Rate]
\label{ass:rate}
There exists a sequence $\epsilon_N = o(1)$, such that with probability at least $1-\epsilon_N$, for all $k \in [K]$, the first stage nuisance estimate $\widehat{\nu}_k (\cdot)$ belongs to a shrinking neighborhood of its true value $\nu_0(\cdot)$, denoted by $\mathcal{T}^{\nu}_N$. Uniformly over $\mathcal{T}^{\nu}_N$, the following worst-case rate bound holds. 
$$
\sup_{ \nu \in \mathcal{T}^{\nu}_N} \sup_{x \in \mathcal{X}} \| \nu(x) - \nu_0(x) \| \leq \nu^{\infty}_N = o (N^{-1/4}).
$$
\end{assumption}

When the convergence is required in mean square ($\ell_2$) norm, Assumption \ref{ass:rate} is a classic assumption in the semiparametric literature (e.g., \citep{Newey1994}). The example below demonstrates the plausibility of Assumption \ref{ass:rate} in $\ell_{\infty}$ norm.

\begin{continuance}{\ref{ex:welfare}}
Let $\mathcal{D} = \{0,1\}$.  Suppose regression functions in Example \ref{ex:welfare} obey linear index restrictions
\begin{align*}
m(1,X)= X^{\prime} \gamma_1, \quad m(0,X) = X^{\prime} \gamma_0. 
\end{align*}
Then, the $\ell_1$-regularized estimator  of \citep{Program} provides a convergence rate bound in $\ell_1$-norm in terms of the sparsity index, which suffices for a uniform rate bound on functions $m(1,x)$ and $m(0,x)$.
\end{continuance}

\begin{assumption}[Regularity Conditions] 
\label{ass:reg} The following technical conditions hold. (1) The moments  are uniformly bounded a.s.
\begin{align}
\label{eq:boundedsm3}
\sup_{ \xi \in \Xi_N}\sup_{t \in \mathcal{T}} \sup_{x \in \mathcal{X}}  | \rho(W,t,\xi) | \leq B_{\rho}.
\end{align}
(2) The bounded derivative condition holds:
\begin{align}
\label{eq:boundedsm2}
\sup_{\nu\in\mathcal{T}_N^\nu} \sup_{x\in\mathcal{X}}\sup_{t \in \mathcal{T}} \| \partial \phi (t, v(x))/ \partial v  \|  \leq B_{\phi}.
\end{align}
\end{assumption}

Assumption \ref{ass:reg} is a standard regularity condition\begin{footnote} {The proof of Theorem \ref{thm:closedform} only requires that the conditional second moment is uniformly bounded
$
\label{eq:boundedsm}
\sup_{ \xi \in \Xi_N}\sup_{t \in \mathcal{T}} \sup_{x \in \mathcal{X}} \E [ \rho^2(W,t,\xi) \mid X=x] \leq B_{\rho}.
$
However, the a.s. bound \eqref{eq:boundedsm3}  is easier to verify in  our examples which is why it is chosen for Assumption \ref{ass:reg}.} \end{footnote}. For example, for the case of Example \ref{ex:welfare},  the condition \eqref{eq:boundedsm2} automatically holds with $B_{\phi}=1$ since $\phi (d, v)=-\mathbf{e}_{d} v$ corresponds to taking the $d$'th element of the nuisance parameter $\nu_0(x) = (m(d,x))_{d \in \mathcal{D}}$.

\begin{assumption}[Margin Assumption]
\label{ass:ma}
There exist finite positive constants $\bar{B}, \delta \in (0, \infty)$  such that
\begin{align}
\label{eq:ma}
\sup_{(j,k) \in \mathcal{T}, \quad k \neq j} \Pr \left( 0 \leq  \phi (j, \nu_0(X))-  \phi (t, \nu_0(X)) \leq t \right) \leq \bar{B} t, \quad \forall t \in (0, \delta).
\end{align}
\end{assumption}

Assumption \ref{ass:ma} is  a margin condition that ensures separation between minimizers and non-minimizers in order to control the  first-order effect of classification mistakes. It is a standard assumption in classification literature  \citep{MammenTsybakov,Tsybakov,QianMurphy}, policy learning \citep{KitagawaTetenov,MbakopTabord} and debiased inference \citep{kallus2022whats,SemJoE,SemSupp2}.  Section \ref{sec:ma} verifies Assumption \ref{ass:ma} for a special case of single-index models.

\subsection{Discussion of  Assumption \ref{ass:ma}}
\label{sec:ma}

In this section, I demonstrate the plausibility of Assumption \ref{ass:ma} in the context of Example \ref{ex:welfare}.  

\begin{remark}[Binary Treatment]
\label{lem:binary}
Suppose $\mathcal{D} = \{0,1\}$.  If the conditional average treatment effect $m(1,X) -m(0,X)$ has a bounded density
\begin{align}
\label{eq:margin2}
\exists \delta>0, \ s.t. \sup_{t \in (-\delta, \delta) } f_{ m(1,X) - m(0,X)} (t) \leq \bar{B}_f,
\end{align}
then  Assumption \ref{ass:ma} is satisfied with $\bar{B}= 2 \bar{B}_f$.
\end{remark}

Remark  \ref{lem:binary} verifies that Assumption \ref{ass:ma} holds with $\bar{B} = 2 \bar{B}_f$ as long as the index set $\mathcal{T}$ has two elements and their difference $m(1,X)-m(0,X)$ has a bounded unconditional density. This is a known result in the literature  \citep[e.g.,][]{Tsybakov,KitagawaTetenov,kallus2022whats}.

Lemma \ref{ex:1} describes a class of linear models where the presence of covariate vector $X$ with a smooth distribution suffices for Assumption \ref{ass:ma}. Suppose the  expectation functions are partially linear 
\begin{align}
\label{eq:linearmodel}
m(t, \widetilde{X}) = X^{\prime} \gamma_t + g_t(\bar{X}), \quad t \in \mathcal{T}.
\end{align}
where $\widetilde{X} = (X, \bar{X})$ and $ \bar{X} $ is independent  of $X$.
\begin{lemma}[Linear Model]
\label{ex:1}
 Suppose
\begin{enumerate}
\item[(i)] For any $j,k \in \mathcal{T}$, $\gamma_k \neq \gamma_j$. 
    
      \item[(ii)] For some $M < \infty$, $\max_{t \in \mathcal{T}} |g_t(\bar{X})| \leq M$ almost surely.
      
      \item[(iii)]  The vector $X$ obeys a smoothness condition
      \begin{align}
\label{eq:macov}
\sup_{\delta \in \mathrm{R}^{p_X}, \| \delta \|=1} \Pr \left( 0 <  X^{\prime} \delta < t \right) \leq \bar{B} t, \quad t \rightarrow 0.
\end{align}

\end{enumerate}

Then, Assumption \ref{ass:ma} is satisfied.
\end{lemma}

Lemma \ref{ex:1} verifies Assumption \ref{ass:ma} as long as the model \eqref{eq:linearmodel} includes a linear component obeying \eqref{eq:macov}. Condition (i) ensures that the mapping $x \mapsto \min_{t \in \mathcal{T}} (x^{\prime} \gamma_t)$ has a unique minimum. Condition (ii) accommodates the inclusion of arbitrary covariates (i.e., either continuous or discrete), provided their influence on the conditional mean remains almost surely bounded. Condition (iii) is a smoothness condition on $X$ similar to those in the analysis of  least absolute deviation  in \cite{Powell}  or the analysis of support function in  \cite{CCMS}. For example, if $X \sim N(h, \Sigma)$ is a Gaussian $p_X$-vector, the scalar $X^{\prime} \delta \sim N(\delta^{\prime} h, \delta^{\prime} \Sigma \delta)$, and the condition \eqref{eq:macov} holds with $ \bar{B} = \sqrt{2 \pi \delta^{\prime} \Sigma \delta}^{-1} \leq (\sqrt{ 2 \pi})^{-1} \lambda^{-1/2}_{\min} (\Sigma)$.

Lemma \ref{ex:3} extends the result of Lemma \ref{ex:1} to nonlinear models. Suppose the  expectation functions are
\begin{align}
\label{eq:linearmodel2}
m(t, X) = F(X^{\prime} \gamma_t), \quad t \in \mathcal{T}.
\end{align}

\begin{lemma}[Nonlinear Model]
\label{ex:3}
  Suppose
\begin{enumerate}
     
        \item[(i)] Conditions (i)-(iii) of Lemma \ref{ex:1} hold.
  
     \item[(ii)] The covariate vector $X$ is $B_X$-bounded, and the support of $\cup_{t \in \mathcal{T}} X^{\prime} \gamma_t$, denoted by $ \mathcal{X}$, is compact.
     
    \item[(ii)]  The link function derivative is bounded from below on $\mathcal{X}$
     \begin{align*}
   \inf_{ t \in \mathcal{X}}  \frac{dF(t)}{dt} \geq \underline{f} >0.
\end{align*}

\end{enumerate}
Then, Assumption \ref{ass:ma} is satisfied.
\end{lemma}

Lemma \ref{ex:3} verifies Assumption \ref{ass:ma} for a single-index model with a monotone link function provided the linear index obeys \eqref{eq:ma}. Condition (iii) is satisfied for a large class of link functions, such as probit, logit or uniform $U[-t, t]$ where $\mathcal{X}$ is included in $[-t,t]$.

In conclusion, let me point out that Assumption \ref{ass:ma} can be viewed as a special case of a general form of the margin assumption
\begin{align}
\label{eq:ma2}
\sup_{(j,k) \in \mathcal{T}, \quad k \neq j} \Pr \left( 0 \leq  \phi (j, \nu_0(X))-  \phi (t, \nu_0(X)) \leq t \right) \leq \bar{B} t^{\widetilde{\alpha}}, \quad \forall t \in (0, \delta), \quad \text{ for some } \widetilde{\alpha} \in (0,1]
\end{align}
that is routinely imposed in standard debiased inference \citep[e.g.,][]{LuedtkeLaan,kallus2020assessing,SemSupp2}.  Relaxing this assumption remains an important open question in the literature. When this assumption is violated, the cross-fit plug-in estimators have a non-standard, heavy-tailed distribution which makes standard Wald-type inference not valid \citep{LuedtkeLaan,ponomarev2024}.

\subsection{Asymptotic Results}
\label{sec:res}

\begin{theorem}[Asymptotic Theory]
\label{thm:closedform}
Under Assumptions \ref{ass:smallbias}--\ref{ass:reg}, the proposed estimator obeys the oracle property
\begin{align}
\label{eq:oracle}
\sqrt{N} ( N^{-1} \sum_{i=1}^N \rho(W_i, \widehat t_i,  \widehat \xi_i)    -N^{-1} \sum_{i=1}^N \rho(W_i, t_0, \xi_0)) = o_P(1).
\end{align}  
Therefore, it is asymptotically Gaussian 
$$
\sqrt{N} (N^{-1} \sum_{i=1}^N \rho(W_i, \widehat t_i,  \widehat \xi_i) - \psi_0) \Rightarrow^d N(0, V_0),
$$
with the asymptotic variance in \eqref{eq:psieta0}.
\end{theorem}

Theorem \ref{thm:closedform}   is my first main result\begin{footnote}{The oracle property is reminiscent of Neyman orthogonality in the double/debiased machine learning literature \citep{Neyman:1959,Neyman:1979,chernozhukov2016double}. In the context of Example \ref{ex:welfare}, the derivative calculation appears in the calculation of the efficiency bound (Lemma \ref{lem:LvdL}, see the proof of Theorem 1 in Supplement).  }\end{footnote}. As a special case, it nests recent debiased inference results for sharp Makarov bounds with a binary outcome \citep{kallus2022whats}  and for \cite{BalkePearl1994,BalkePearl1997} bounds in the concurrent, independent work of  \cite{levis2023covariateassisted}. The paper's contribution is to introduce a general framework of aggregated intersection of regression functions for which the oracle property also applies. The new applications of the framework include Roy model bounds and Horowitz-Manski-Lee bounds with discrete outcomes, discussed in Section \ref{sec:appl}. 

\begin{theorem}[Consistent Estimation of Asymptotic Variance]
\label{thm:consistency}
Under Assumptions \ref{ass:smallbias}--\ref{ass:reg}, the sample analog estimator $\widehat V$ in \eqref{eq:varhat} is consistent for $V_0$ in \eqref{eq:psieta0}, that is $\widehat V - V_0 = O_P (\nu^{\infty}_N + \Lambda_N^{1/2}) = o_P(1)$. \end{theorem}

Theorem \ref{thm:consistency} establishes consistency of the plug-in estimator of variance which suffices for the validity of the confidence interval \eqref{eq:plugin}.  Unlike the envelope score estimator $\widehat \psi$, the variance estimator $\widehat V$ is first-order sensitive to the mistakes in estimated minimizers and converges at rate $\nu^{\infty}_N$ rather than $(\nu^{\infty}_N)^2$.

\begin{theorem}[Bootstrap Inference]
\label{thm:boot}
Under Assumptions \ref{ass:smallbias}--\ref{ass:reg}, for $\widehat{c}_N(1-\alpha)$ being the $(1-\alpha)$-quantile of $ \widetilde{S}_N= \sqrt{N} (\widetilde \psi - \widehat \psi)$ under $P^e$,
 \begin{align*}
\Pr ( \sqrt{N} (\widehat \psi - \psi_0) \leq \widehat{c}_N(1-\alpha)) \rightarrow 1-\alpha,
\end{align*}
which implies \eqref{eq:coverageboot}.
\end{theorem}

Theorem \ref{thm:boot} establishes the validity of multiplier bootstrap inference. Note that the bootstrap-based inference is possible because the kink point (the point of non-differentiability) occurs with probability zero. A similar validity argument is used to establish bootstrap inference for support function as in \cite{CCMS} and \cite{SemJoE}.

\section{Applications}
\label{sec:appl}

\subsection{Roy model with binary outcome}
\label{sec:roy}

I begin by reviewing Roy model.   Let $Y(1)$ and $Y(0)$ be two binary potential utility values, corresponding to the choices of treatment $D=1$ and $D=0$, respectively. An individual chooses the treatment value $D$ according to
\begin{align}
\label{eq:choice}
Y(1) > Y(0) \Rightarrow D=1, \quad Y(1) < Y(0) \Rightarrow D=0.
\end{align}
If the potential outcomes $Y(1)$ and $Y(0)$ are equal, the choice is unspecified.  The observed data $W=(X,D,Y)$ consist of the covariates $X$, the choice variable $D$, and the observed outcome $Y=DY(1) + (1-D) Y(0)$.

\begin{proposition}[Proposition 1,  \citep{MourifieHenry}]
\label{prop:mh}
Suppose there exists a vector $Z$ such that it satisfies the exogeneity restriction
\begin{align}
\label{eq:y1y0}
(Y(1), Y(0)) \ci Z.
\end{align}
Then,  the sharp bounds on the joint distribution of potential outcomes are
 \begin{align}
\Pr (Y(1)=1, Y(0) =0) &\leq \min_{z \in \mathcal{Z}} \Pr (Y=1, D=1\mid Z=z), \label{eq:roy2} \\
\Pr (Y(1)=0, Y(0) =1) &\leq  \min_{z \in \mathcal{Z}} \Pr (Y=1, D=0\mid Z=z).  \label{eq:roy3}   
\end{align}
\end{proposition}

Proposition \ref{prop:mh} restates the Proposition 1 of \citep{MourifieHenry}. It gives sharp bounds in Roy model with an instrument $Z$ obeying exclusion restriction. Such variables are akin to typical instrumental variables. The examples of discrete-valued $Z$ provided in \citep{MourifieHenry} include parental education, distance to a college, and attendance a Catholic high school.

\begin{assumption}[Instruments and Covariates]
\label{ass:commsupp}
(A) Conditional independence. The discrete-valued instrument $Z$ is independent of the potential outcomes conditional on $X$
 \begin{align}
 \label{eq:condindep2}
(Y(1),Y(0)) \ci Z \mid X.
 \end{align}
 (B) Complete independence.  The discrete-valued instrument $Z$ is independent of the potential outcomes
 \begin{align}
 \label{eq:uncondindep}
(Y(1),Y(0),X) \ci Z.
 \end{align}
 \end{assumption}

Assumption \ref{ass:commsupp} (A) is a relaxation of \eqref{eq:y1y0} which requires that independence holds only conditional on covariates. Assumption \ref{ass:commsupp} (B) is equivalent to \eqref{eq:y1y0}. 

\begin{proposition}[Sharp bounds in Roy model with covariates]
\label{prop:roy2}
(A) Suppose Assumption \ref{ass:commsupp}(A) holds. Then, the sharp bounds on the joint distribution of potential outcomes are aggregated intersection bounds
 \begin{align}
\Pr (Y(1)=1, Y(0) =0) &\leq \E [\min_{z \in \mathcal{Z}} \Pr (Y=1, D=1\mid Z=z, X)]  \label{eq:roy22} \\
\Pr (Y(1)=0, Y(0) =1 ) &\leq \E [ \min_{z \in \mathcal{Z}} \Pr (Y=1, D=0\mid Z=z, X)] \label{eq:roy23}   
\end{align}
(B) Jensen's inequality implies 
\begin{align}
\E [\min_{z \in \mathcal{Z}} \Pr (Y=1, D=1\mid Z=z, X)] \leq \min_{z \in \mathcal{Z}} {\Pr}_{\mu} (Y=1, D=1\mid Z=z),  \label{eq:roy32}  \\
\E [ \min_{z \in \mathcal{Z}} \Pr (Y=1, D=0\mid Z=z, X)] \leq \min_{z \in \mathcal{Z}} {\Pr}_{\mu} (Y=1, D=0\mid Z=z),  \label{eq:roy33}
\end{align}
where 
\begin{align}
{\Pr}_{\mu} (Y=1, D=d\mid Z=z) := \dfrac{ \E \bigg[ \dfrac{ 1\{ D=d \} \cdot Y \cdot 1\{ Z=z \}}{ \mu_{z0}(X)} \bigg]}{ \Pr (Z=z)}, \quad z \in \mathcal{Z}.
\end{align}
(C)  Furthermore, if the propensity score is constant in $X$, 
$$
\mu_{z0}(X) = \Pr (Z=z), \quad z \in \mathcal{Z}, \quad \text{ a.s. },
$$
\eqref{eq:roy32}--\eqref{eq:roy33} reduce to regular bounds in Proposition \ref{prop:mh}
 \begin{align}
{\Pr}_{\mu} (Y=1, D=1\mid Z=z) =   \Pr (Y=1, D=d\mid Z=z), \quad \forall d \in \{0, 1\}.  \label{eq:roy34}   
\end{align}
\end{proposition}

Proposition  \ref{prop:roy2}  refines Proposition \ref{prop:mh} by incorporating covariate information. The sharp bounds  are derived by  applying the argument in Proposition \ref{prop:mh}, conditional on  $X$, and then aggregating over the covariate space.
Interchanging expectation and minimum gives another pair of bounds of the form \eqref{eq:roy2}--\eqref{eq:roy33}. Since they do not involve any expectation functions and are simpler to estimate, we refer to them as basic (or no-covariate) bounds. If the propensity score is constant, these basic bounds coincide with the original bounds defined in Proposition \ref{prop:mh}.

Let me demonstrate the proposed inferential methodology focusing on the first bound in \eqref{eq:roy22}. The bound 
\begin{align}
\label{eq:psi00}
\psi_0 =  \E [\min_{z \in \mathcal{Z}} \Pr (Y \cdot D =1\mid Z=z, X)]
\end{align}
is a special case of \eqref{eq:mainpsi} with $\mathcal{T} = \mathcal{Z}$, the nuisance vector-function
$$
\nu_0(x) = (\Pr (D =1,  Y=1 \mid Z=z, X=x))_{z \in \mathcal{Z}},
$$
and the projection functions
$$
\phi(z, v) = v_z \quad z \in \mathcal{Z}.
$$
Furthermore, it can be also mapped to the optimal welfare parameter $\psi_0$ in Example \ref{ex:welfare} with 
\begin{align}
\label{eq:mapping}
\mathcal{D}:=\mathcal{Z}, \quad   D:=Z, \quad Y:=(D\cdot Y).
\end{align}
Following Example \ref{ex:welfare}, define the orthogonal score for each $\phi(z,\nu_0(x))$ as
\begin{align*}
\rho(W, z, \xi) &=  \Pr (D =1,  Y=1 \mid X, Z=z) + \dfrac{ 1\{ Z=z\} }{\mu_z(X)} \left(D \cdot Y- \Pr (D =1,  Y=1 \mid X, Z=z) \right),
\end{align*}
where the true value $\xi_0$ of the nuisance parameter is
$$
\xi_0(x) = (\nu_0(x), \mu_{z0}(x)), \qquad \mu_{z0}(x) = \Pr (Z=z \mid X=x), \quad z \in \mathcal{Z}.
$$
Given the true functions $\nu_{z0}(\cdot), \mu_{z0}(\cdot)$ and  sequences of shrinking neighborhoods $\mathcal{T}_N^z$ of  $\nu_{z0}(\cdot) $ and  $M_N^z$  of $\mu_{z0}(x)$, define the following rates:
	\begin{align*}
	 \nu^{\infty}_N := \sup_{z \in \mathcal{Z}} \sup_{\nu_{z} \in \mathcal{T}^z_N} \sup_{x \in \mathcal{X}} | \nu_{z} (x) - \nu_{z0}(x) |, \\
	  \nu_N := \sup_{z \in \mathcal{Z}} \sup_{\nu_{z} \in \mathcal{T}^z_N} ( \E   ( \nu_{z} (X) - \nu_{z0} (X) )^2)^{1/2}, \\
	 \mu_N := \sup_{z \in \mathcal{Z}} \sup_{\mu_{z} \in \mathcal{M}_N^z} ( \E   ( \mu_{z} (X) - \mu_{z0} (X) )^2)^{1/2}.
	\end{align*}

Assumption \ref{ass:fsrate} states the regularity conditions. First, it requires the instrument $Z$ to be discrete-valued with finite support so that  the propensity score
\begin{align}
\label{eq:so}
\kappa \leq \mu_{z0} (x) \leq 1-\kappa, \quad \forall x \in \mathcal{X} \forall z \in \mathcal{X}
\end{align}
is bounded away from zero and one for each distinct instrument value.  Continuously supported instrument (i.e.,  $\mathcal{T}=\mathcal{Z}$) is outside of the scope of this paper since their propensity score violates \eqref{eq:so}. In this case, I conjecture that the propensity score  needs to be approximated by a kernel density estimator, e.g. as is standard in the results for continuous treatments   \citep{Colangelo}.  Second, the mean square rates of the first-stage estimators must decay sufficiently fast, a condition that is standard in semiparametric estimation literature.

 \begin{assumption}[Regularity Conditions for Roy Model with Covariates] 
\label{ass:fsrate} 
Assume that there exists a sequence of numbers $\epsilon_N = o(1)$ and sequences of neighborhoods $\mathcal{T}^z_N$ of  $\nu_{0z}(\cdot) $ and $M_N^z$ of $\mu_{0z}(x)$ such that both the true value $\xi_0(x)$ and the first-stage estimate $\{\widehat{\nu}_z(\cdot), \widehat{\mu}_z(\cdot) \} $ belong to the set $\{\mathcal{T}^z_N \bigtimes M_N^z\}$ w.p. at least $1-\epsilon_N$ for each $ z \in \mathcal{Z}$. The functions in  $M_N^z$ are bounded uniformly over their domain from above and below by $\kappa$ and $1-\kappa$. Finally, assume  that mean square rates $\nu_N,\mu_N$ decay sufficiently fast: $$N^{1/2} \nu_N\mu_N  = o(1), \quad \nu_N \vee\mu_N = o(1), \quad \nu_N^{\infty} = o(N^{-1/4}).$$ 
\end{assumption}
 
Corollary \ref{cor:roy} gives an envelope score estimator for  Roy model bounds and delivers uniformly valid debiased inference  in the presence of covariates. As described in \eqref{eq:mapping}, the sufficient conditions for the margin assumption discussed in Section \ref{sec:ma} are equally applicable to Roy model bounds.

\begin{corollary}[Asymptotic Theory for Roy Model with Covariates]
\label{cor:roy}
Suppose Assumptions \ref{ass:commsupp} (A) and \ref{ass:fsrate} hold, and Assumption \ref{ass:ma}  holds for $\nu_0(X) = (\Pr (D =1,  Y=1 \mid Z=z, X))_{z \in \mathcal{Z}}$.  Then, the statements of Theorems \ref{thm:closedform}--\ref{thm:boot} hold for the estimator of Definition \ref{def:envscore} with $\psi_0$ in \eqref{eq:psi00}.
\end{corollary}

\subsection{Horowitz-Manski-Lee bounds with discrete outcomes.}
\label{sec:lee}

I begin by introducing the sample selection problem.  Let $D=1$ be an indicator for treatment receipt.  Let $Y(1)$ and $Y(0)$ denote the  potential outcomes if an individual is treated or not, respectively.  Likewise, let $S(1)=1$ and $S(0)=1$ be  dummies for whether an individual's outcome is observed  with and without treatment, respectively.  The  data vector $W=(D,X,S,S \cdot Y)$  consists of the treatment status $D$, a baseline covariate vector $X$, the observed selection status $S=D \cdot S(1) + (1-D) \cdot S(0)$ and the observed outcome $S \cdot Y = S \cdot (D  \cdot Y(1) + (1-D) \cdot Y(0))$ for selected individuals.  \citep{LeeBound} focuses on the average treatment effect (ATE)
 \begin{align}         
 \label{eq:truebeta}
 \beta_0 = \E[Y(1) - Y(0)  \mid S(1)=1, S(0)=1] 
 \end{align} 
 for subjects who are selected into the sample regardless of treatment receipt---the \emph{always-takers}.   
 
 \begin{assumption}[Assumptions of \citep{LeeBound}]
 \label{ass:identification:treat}
The following statements hold. 
\begin{compactenum}[(1)]
\item[(1)] (Independence). The  vector $(Y(1),Y(0),S(1),S(0))$ is independent of $D$ conditional on $X$. The propensity score 
  \begin{align}
 \label{eq:propscore}
 \mu_{10}(X) := \Pr (D =1 \mid X), \qquad \mu_{00}(X) =1 - \mu_{10}(X) 
 \end{align} 
 is assumed known.
 	\item[(2)] (Monotonicity).  \begin{equation} \label{eq:monot1} S(1) \geq S(0) \quad \text{ a.s. }. \end{equation}
\end{compactenum}
 \end{assumption}
 
 Assumption \ref{ass:identification:treat}(1) holds by random assignment. Assumption \ref{ass:identification:treat}(2) states that all subjects must exhibit the same direction of selection response. It is frequently imposed in selection and treatment choice models. If covariates are available, this assumption has testable implications. Furthermore, it can be relaxed to conditional monotonicity \citep{Kolesar,SemSupp2}.  In this paper, we consider an unconditional version of the monotonicity assumption so as to focus on the theory for discrete-valued outcomes.  As discussed in \citep{LeeBound}, the average control outcome is point-identified
 $$ \E [ Y(0)  \mid  S(0)=1] = \E [ Y(0)  \mid S(1)=1, S(0)=1] =\E [ Y \mid S=1, D=0].$$
I focus on the average treated outcome 
 \begin{align}
 \label{eq:truebeta2}
 \beta_1 = \E[Y(1)  \mid S(1)=1, S(0)=1].
 \end{align} 
In contrast to the control group, a treated outcome can be either an always-taker's outcome or a complier's outcome. The always-takers' share among the treated outcomes is
\begin{align}
\label{eq:p0}
p_0 = \Pr[ S(1) =1, S(0)=1\mid S(1) =1 ]= \Pr[ S(0)=1\mid S(1) =1 ] =  \dfrac{\Pr [S=1\mid D=0]}{\Pr [S=1\mid D=1]} =: \dfrac{s_0}{s_1}.
\end{align}

\paragraph{Binary outcomes. } Suppose $Y(1)$ and $Y(0)$ take values in $\{1, 0\}$. In the best case, the always-takers comprise the top $p_0$-quantile of the treated outcomes. Let $p_Y:= \Pr (Y =1 \mid D=1, S=1)$.  In case when $p_0 <p_Y$, 
 the best-case  always-takers' outcome is equal to one for all always-takers. Otherwise, the best-case always-takers' outcome distribution is a mixture of ones and zeroes, with the mixing proportion of ones and zeroes  equal to $p_Y/p_0$ and $1-p_Y/p_0$, respectively.  In other words, the basic (i.e., no-covariate) upper bound $\bar{\beta}_U$ on $\E[ Y(1) \mid S(1) =1, S(0)=1]$ can be expressed as
\begin{align}
\label{eq:basicupper}
\bar{\beta}_U &= \min (p_Y/p_0 -1, 0)+1 = \min \left( \dfrac{s_1 p_Y -  s_0 }{s_0}, 0 \right) +1.
\end{align}
Since $\bar{\beta}_U$ involves no covariates, we refer to it as the basic bound as opposed to the sharp bound derived further.

Lee's identification strategy can be implemented conditional on covariates. Denote the conditional trimming threshold $p_0(x)$ as
 \begin{align}
  p_0(x) = \dfrac{\Pr (S=1\mid  D=0, X=x)}{\Pr (S=1\mid D=1, X=x)}=\dfrac{s_0(0,x)}{s_0(1,x)} \quad x \in \mathcal{X} \label{eq:condtrim}
 \end{align}
 and the conditional probability of outcome one as 
 $$
 p_Y(x) = \Pr (Y = 1\mid D=1, S=1, X=x). 
 $$
Finally, the conditional upper bound $\bar{\beta}_U(x)$ can be expressed as
   \begin{align}
   \label{eq:condbetau}
\bar{\beta}_U (x)= \min (p_Y(x)/p_0(x) -1, 0 ) +1.
\end{align}
Aggregating the conditional bound over the always-takers' covariate distribution gives the sharp upper bound
   \begin{align}
\beta_U &= \int_{x\in \mathcal{X}} \bar{\beta}_U (x) f_{X}(x \mid S(0)=1, S(1)=1) dx = \dfrac{\E [  \bar{\beta}_U (X) s_0(0,X) ] }{ \E [ s_0(0,X) ]} \label{eq:sharpbetau},
 \end{align}
 and a similar argument applies for the lower bound. Proposition \ref{prop:lee} derives the sharp lower and upper bounds for the average potential outcome.
 
 \begin{proposition}
\label{prop:lee}
 Suppose Assumption \ref{ass:identification:treat} holds for a binary outcome $Y$ taking values in $\{0, 1\}$. Then, the following statements hold:
 (A) The sharp lower and upper bounds on $\beta_1$ in \eqref{eq:truebeta2} take the form of ratios
\begin{align}
\label{eq:leebounds}
 \beta_L = \dfrac{N_L}{ \E [ s_0(0,X) ] }, \qquad  \beta_U = \dfrac{N_U}{  \E [ s_0(0,X) ]  },
 \end{align}
whose numerators are aggregated intersection bounds 
  \begin{align}
 N_L &=  \E [ \max (s_0(0,X) - s_0(1,X)(1-p_Y(X)), 0) ] , \label{eq:nlbinary} \\
  N_U &= \E  [\min (s_0(1,X) p_Y(X)-s_0(0,X), 0) +s_0(0,X) ] . \label{eq:nubinary} 
 \end{align}
 (B)  Jensen's inequality implies
  \begin{align}
\E [ \max (s_0(0,X) - s_0(1,X)(1-p_Y(X)), 0) ] &\geq \max (s_0 -  s_1 (1-p_Y) , 0),  \label{eq:nlbinary2} \\
\E  [\min (s_0(1,X) p_Y(X)-s_0(0,X), 0)  ] &\leq \min ( p_Y s_1 -  s_0, 0). \label{eq:nubinary2} 
 \end{align}

\end{proposition}

Proposition \ref{prop:lee} derives   basic and sharp bounds on the average potential outcome in a selection problem. The denominators of basic and sharp bounds are the same and equal 
the always-takers' share $$s_0 = \Pr [S=1 \mid D=0] = \E [ s_0(0,X) ].$$ Their numerators are regular and aggregated intersection bounds, described in the LHS and RHS of \eqref{eq:nlbinary2}--\eqref{eq:nubinary2}, respectively.  The basic bounds \eqref{eq:nlbinary2}--\eqref{eq:nubinary2} coincide (up to a constant) the bounds in the Lemma 1 of concurrent, independent work of \citep{kroft2024leeboundsmultilayeredsample}.

\paragraph{Discrete outcomes. }In this section, I allow the outcome $Y$ to take a finite number of discrete values. Proposition \ref{prop:leed} characterizes the numerators of Horowitz-Manski-Lee bounds as a special case of aggregated intersection bounds.

\begin{proposition}[Horowitz-Manski-Lee Bounds with Discrete Outcomes]
\label{prop:leed}
 Suppose Assumption \ref{ass:identification:treat}  holds with a discrete outcome $Y$ whose support is denoted by $\mathrm{T}$.   (A) Then, the sharp lower and upper bound on $\beta_1$ are given in \eqref{eq:leebounds} with $N_L$ and $N_U$ given in 
   \begin{align}
  \label{eq:nl}
 N_L &= \E [\max_{\beta \in \mathrm{T} } (\beta s_0(0,X) + s_0(1,X) \E [ \min (Y-\beta,0)\mid D=1, S=1, X])]
 \end{align}
 and the upper bound numerator
   \begin{align}
   \label{eq:nu}
N_U &=  \E [ \min_{\beta \in \mathrm{T} } ( \beta s_0(0,X) +  s_0(1,X) \E[ \max (Y-\beta,0)\ \mid D=1, S=1, X])].
\end{align}
(B) The basic bounds on $\beta_1$ are given in \eqref{eq:leebounds} with $N_L$ and $N_U$ given in 
   \begin{align*}
 N_L &= \max_{\beta \in \mathrm{T} } (\beta s_0 + s_1 \E [ \min (Y-\beta,0)\mid D=1, S=1])
 \end{align*}
 and the upper bound numerator
   \begin{align*}
N_U &=  \min_{\beta \in \mathrm{T} } ( \beta s_0 +  s_1 \E[ \max (Y-\beta,0) \mid D=1, S=1]).
\end{align*}

 \end{proposition}
 
 Proposition \ref{prop:leed} develops a novel representation of Horowitz-Manski-Lee bounds as regular and aggregated intersection bounds, respectively. The outcome distribution is represented using point mass functions (PMFs) rather than quantiles, which is convenient for working with discrete outcomes. As shown in \citep{RockUryasev}, the minimum in \eqref{eq:nu} is attained by the outcome quantile of level $1-s_0/s_1$ or the ``borderline'' always-takers' outcome.

\paragraph{Debiased inference. } To describe an inferential approach,  I derive moment functions for \eqref{eq:nl} and \eqref{eq:nu}. Define the moment functions for the lower bound
   \begin{align}
     \label{eq:glhelp_beta}
\rho_L (W, \beta):&=   \dfrac{DS(Y- \beta )}{\mu_{11}(X) }   \cdot 1{\{ Y \leq \beta     \}} +\beta  \dfrac{(1-D) S}{\mu_{00}(X) },
\end{align}
and for the upper bound
   \begin{align}
     \label{eq:guhelp_beta}
\rho_U (W, \beta):&=   \dfrac{DS(Y- \beta )}{\mu_{11}(X) }   \cdot 1{\{ Y \geq \beta     \}} +\beta  \dfrac{(1-D) S}{\mu_{00}(X) }.
\end{align}
Next, let $ \mathrm{T}$ be the finite support of the outcome $Y$.  Define the nuisance parameter 
\begin{align*}
\xi_0(x)= \{ s_0(0,x), s_0(1,x),{\pi_{\beta 0}(x)}_{\beta  \in  \mathrm{T} } \}.
\end{align*}
The resulting moment function for $N_U$ is 
\begin{align}
g_U(W, \xi) = \sum_{\beta \in \mathrm{T} } \rho_U (W, \beta) 1\{ \beta = \arg \min_{\beta \in \mathrm{T} } \beta s(0,X) + s(1,X) \E [ \max (Y -\beta, 0) \mid D=1, S=1, X] \}.
\end{align}
 Likewise, the moment function for  $N_L$ is 
\begin{align}
g_L(W, \xi) = \sum_{\beta \in \mathrm{T} } \rho_L (W, \beta) 1\{ \beta = \arg \max_{\beta \in \mathrm{T} } \beta s(0,X) + s(1,X) \E [ \min (Y -\beta, 0) \mid D=1, S=1, X] \}.
\end{align}
Given the true functions $\xi_{0}(\cdot)$ and  sequences of shrinking neighborhoods $S^d_N$ of  $s_0(d,x)$ and $\mathcal{P}^{\beta}_N$ of $\pi_{\beta 0}(x)=\Pr (Y = \beta \mid D=1, S=1, X=x)$,  define the following rates:
	\begin{align*}
	 s^{\infty}_N :&= \sup_{d \in \{1, 0\}} \sup_{s(d, \cdot) \in \mathcal{S}^d_N} \sup_{x \in \mathcal{X}} | s(d,x) -s_0(d,x) |, \\
		   \pi^{\infty}_N :&= \sup_{\beta \in \mathrm{T}} \sup_{\pi_{\beta} \in \mathcal{P}^{\beta}_N} \sup_{x \in \mathcal{X}} | \pi_{\beta} (x) -\pi_{\beta0} (x)  |  
	\end{align*}
	
 \begin{assumption}[Regularity Conditions for Horowitz-Manski-Lee Bounds] 
\label{ass:fsrate2} 
Assume that there exists a sequence of numbers $\epsilon_N = o(1)$ and sequences of neighborhoods $\mathcal{P}^{\beta}_N$ of  $\pi_{\beta0} (\cdot) $ and $\mathcal{S}^d_N$ of $s_0(d,\cdot)$ such that both the true value $\xi_0(x)$ and the first-stage estimate $\{ \widehat{s}(d,\cdot), \widehat{\pi} (\beta, \cdot) \} $ belongs to the set $\{ \mathcal{S}^d_N \bigtimes  \mathcal{P}^{\beta}_N \}$ w.p. at least $1-\epsilon_N$ for each $ \beta \in \mathrm{T}$ and $d \in \{1, 0\}$. 
(i) The rates $  \pi^{\infty}_N$ and $ s^{\infty}_N$ are sufficiently fast: $$  \pi^{\infty}_N + s^{\infty}_N = o(N^{-1/4}).$$ 
(ii) The functions in each set, as well as the propensity score $\mu_{10}(x)$ and  $\pi_{\beta0}(x)$,  are bounded uniformly over their domain from above and below by $\kappa$ and $1-\kappa$. The support set $\mathrm{T}$ is finite.   
(iii) The vector $(s_0(0,X), s_0(1,X),  \cup_{  \mathrm{T}  } \pi_{\beta0}(X) )$ is continuously distributed with a bounded joint density such that each of its component is supported on $(\kappa, 1- \kappa)$. 
   \end{assumption}

Assumption \ref{ass:fsrate2} summarizes regularity conditions for Horowitz-Manski-Lee bounds. Since the propensity score is assumed known, the  individual moment functions $\{ \rho_L (W, \beta), \rho_U (W, \beta)\}_{\{\beta \in \mathcal{T} \}}$ in \eqref{eq:glhelp_beta} and \eqref{eq:guhelp_beta} do not involve any nuisance parameters. As a result, Assumption \ref{ass:smallbias} is automatically satisfied for the individual functions.

\begin{corollary}[Asymptotic Theory for Horowitz-Manski-Lee Bounds with Discrete Outcome]
\label{prop:leed2}
Suppose Assumptions \ref{ass:identification:treat}  and \ref{ass:fsrate2}  hold.  Then, the statements of Theorems \ref{thm:closedform}--\ref{thm:boot} hold for the estimator  described in Algorithm \ref{alg:lee} as well as its bootstrap analog outlined in Definition \ref{def:bb}.  
\end{corollary}

Corollary \ref{prop:leed2} delivers a root-$N$ consistent, asymptotically Gaussian estimator of sharp Horowitz-Manski-Lee bounds assuming the conditional probability of selection and the conditional PMF are estimated at a sufficiently fast rate. To the best of my knowledge, this is a first example of debiased inference for the trimming  bounds with discrete-valued outcome.

In conclusion, I state the Algorithm \ref{alg:lee} for computing the bounds as well as  examples of the first-stage estimators.

\begin{example}[Estimator of Selection Probabilities]
\label{ex:sel}
Suppose the selection probability $s_0(d,x)$ for $d \in \{1, 0\}$ can be approximated by a logistic function
\begin{align}
\label{eq:sel1}
s_0(d,x) = \Lambda ( x' \gamma^d_{0}) + r_d(x), \quad d \in \{1, 0\},
\end{align}
where $\Lambda(\cdot) = \dfrac{\exp (\cdot)}{1 + \exp (\cdot)}$ is the logistic CDF,  $\gamma^d_{0} \in \mathrm{R}^{p}$ is the pseudo-true value of the logistic parameter,  and $r_d(x)$ is its approximation error. The logistic likelihood function is 
\begin{align}
\label{eq:ll}
\ell_d(\gamma^d) =\dfrac{1}{N}  \sum_{i=1}^N (D_i = d ) \bigg(  \log ( 1+ \exp (X_i '\gamma^d))  - S_i X_i'\gamma^d \bigg), \quad d \in \{1, 0\}. 
\end{align}
Given an estimate $\widehat{\gamma}^d$ of $\gamma^d$, define  the estimated selection probabilities as
\begin{align}
\label{eq:sel2}
\widehat{s}(d,x) &=  \Lambda ( x' \widehat {\gamma}^d), \quad d \in \{1, 0 \}
\end{align}
and the estimated CATE on selection
\begin{align*}
\widehat{\tau}(x) =  \widehat{s}(1,x) - \widehat{s}(0,x).
\end{align*}
Given the penalty parameter $\lambda_S$,   the $\ell_1$-regularized logistic estimator of $\gamma^d$ \citep{orthogStructural,Program}  is
\begin{align}
\label{eq:onepenalty}
\widehat{\gamma}^d_{L}= \arg \max_{\gamma^d \in \mathrm{R}^{p}} \ell_d(\gamma^d) + \lambda_S  \| \gamma^d \|_1.
\end{align}

\end{example}

\begin{example}[Estimator of Outcome Probability Mass Function]
\label{ex:outcome}
Suppose the outcome PMF can be approximated by a multinomial logistic regression
\begin{align}
\label{eq:sel1}
\pi_{\beta0}(x) = \Lambda_{\beta}( x' \delta_{\beta 0}) + r_{\beta}(x), \quad \beta \in \mathrm{T},
\end{align}
where $\delta_{\beta 0} \in \mathrm{R}^{p}$ is the pseudo-true value of the logistic parameter and
$$
\Lambda_{\beta}(x' \delta_{\beta 0}) = \frac{\exp(x' \delta_{\beta 0})}{1+\sum_{\beta \in \mathrm{T} \setminus \{0 \}}  \exp(x' \delta_{\beta 0})} ,
$$
 and $r_{\beta}(x)$ is the approximation error. The multiclass classification via sparse multinomial logistic regression is developed in \cite{abramovich2020multiclassclassificationsparsemultinomial}.
 \end{example}

 \begin{algorithm}
 Input: estimated first-stage fitted values $(\widehat\xi_i=\widehat \xi (X_i))_{i=1}^N =(\widehat s(0,X_i), \widehat s(1,X_i), (\widehat \pi_\beta(X_i))_{\beta\in\mathrm{T}})_{i=1}^N$. \\Then, Estimate 
\begin{algorithmic}[1]

\STATE The numerators $N_U$ and $N_L$ 
$$\widehat N_U:=N^{-1} \sum_{i=1}^N g_U(W_i, \widehat{\xi}_i) , \qquad \widehat N_L:= N^{-1} \sum_{i=1}^N g_L(W_i, \widehat{\xi}_i).$$

\STATE The denominator  (the always-takers' share) $$\widehat \pi_{\text{AT}}:=N^{-1} \sum_{i=1}^N g_0(W_i, \widehat{\xi}_i), \text{ where }g_0(W,\xi) = \dfrac{1-D}{\mu_{00}(X)}[S-s(0,X)]+s(0,X)$$

\STATE The preliminary bounds 
\begin{align}
\label{eq:ratio}
\widehat{\beta}_L:=\dfrac{ \widehat N_L }{\widehat \pi_{\text{AT}}},  \quad \widehat{\beta}_U:= \dfrac{\widehat N_U  }{\widehat \pi_{\text{AT}}}
\end{align}
and the  sorted bounds
\begin{align}
\label{eq:genest}
\widetilde{\beta}_L:= \min (\widehat{\beta}_L, \widehat{\beta}_U), \quad \widetilde{\beta}_U:= \max (\widehat{\beta}_L, \widehat{\beta}_U).
\end{align}

\STATE  The $100(1-\alpha)\%$ confidence region is 
\begin{align}
\label{eq:cralpha}
CR^{1-\alpha}  := [ \widehat \beta_L - N^{-1/2} \widehat \Omega^{1/2}_{LL} c_{1-\alpha/2}, \quad \widehat \beta_U +  N^{-1/2} \widehat \Omega^{1/2}_{UU} c_{1-\alpha/2}] 
\end{align}
where the asymptotic covariance matrix $\Omega $ is
\begin{align}
 \Omega = Q \Gamma Q^{T}, \quad Q = ( \pi_{\text{AT}} )^{-1} \begin{pmatrix} 1 & 0 & -\beta_L \\
0 & 1 & -\beta_U  \\
\end{pmatrix} \label{eq:omega}
\end{align}   
and $\Gamma = \text{Var} (g_L(W, \xi_0), g_U(W, \xi_0), g_0(W, \xi_0))$. 

\end{algorithmic}

\caption{ Horowitz-Manski-Lee Bounds}
\label{alg:lee}

\end{algorithm}

\section{Numerical Results}
\label{sec:empirical2}

This section provides numerical evidence for the methods developed in this article. Section \ref{sec:mc} offers a Monte Carlo experiment constructed in the context of Example \ref{ex:welfare}. Section \ref{sec:empirical} offers an empirical illustration of the method for Horowitz-Manski-Lee bounds in Section \ref{sec:lee}.

\subsection{Simulation Study}
\label{sec:mc}

I build a simulation exercise on JTPA dataset \citep{Bloom1997} that consists of three elements: baseline covariates, treatment (access to job training), and outcome.  The baseline covariate $X_1$ is taken to be the previous earnings  \textit{PreEarn} measured in $10, 000$ USD.  The covariate vector $X=(X_1, \dots, X_1^p)$ includes the first $p$ powers of the  \textit{PreEarn} variable. The treatment $D$ is determined by a coin flip with probability $Pr(D = 1) = \frac{2}{3}$, to match the propensity score in JTPA data.  The outcome $Y$ follows a linear model
\begin{align}
\label{eq:sims}
Y = X' \kappa_0 + (DX)' \gamma_0 + \epsilon,
\end{align}
where $\epsilon \sim N(0, \sigma^2)$ is a Gaussian shock independent of the data. The true parameter values are 
$$\kappa_0 = \gamma_0 =(2^{-1}, 2^{-2}, \dots, 2^{-p}), \quad \sigma^2 = 1.$$
The population data set size is $9, 223$. In addition to this primary design, we also consider an artificial (Gaussian) design where  $X_1$  is drawn from a Gaussian distribution whose mean and variance matches the respective parameters of actual  \textit{PreEarn} variable, with all other steps being the same.   Using this setup, we evaluate coverage of plug-in and bootstrap inferential confidence intervals based on the doubly robust estimator described in Example \ref{ex:welfare}. The first-stage functions $m(0, X)$ and $m(1, X)$ are estimated via linear least squares. The performance metrics include bias, mean squared error (MSE), and coverage rates of the confidence intervals (CIs). These metrics are analyzed across varying sample sizes $N \in \{100, 200, 300, 500\}$ and polynomial degrees $p \in \{1, 3, 5, 7\}$.

Tables \ref{tab:welfare} and \ref{tab:welfare2} summarize the simulation results, highlighting key differences across polynomial degrees ($p$) and designs. For lower degrees ($p \in \{1, 3\}$), both designs yield estimators with low bias, low MSE, and near-nominal coverage rates for both plug-in and bootstrap CIs. The results are consistent with the theoretical results in Theorems \ref{thm:closedform} and \ref{thm:boot}.   For higher degrees ($p \in \{5, 7\}$), the performance diverges significantly between the two designs. In the Gaussian design (Table \ref{tab:welfare2}), coverage remains close to the nominal rate even for $p=7$ when $N=500$. Conversely, in the primary design (Table \ref{tab:welfare}), coverage drops sharply to $48\%$ (Plug-In CI) and $39\%$ (Bootstrap CI) for $p=7$. This discrepancy likely arises from the heavy-tailed nature of \textit{PreEarn}, which could either affect the quality of the first-stage estimates of $\kappa_0$ and $\gamma_0$, 
make Assumption \ref{ass:ma} to be a poor fit for the data, or both.

 %
%
%
%
%
%
%

\begin{table}[ht]
\centering
\caption{Finite-Sample Performance of Plug-In and Bootstrap CI}
\label{tab:welfare}
\begin{tabular}{ccccccccc}
 \hline \\
& & & \multicolumn{2}{c}{Coverage} & & & \multicolumn{2}{c}{Coverage} \\
\\
$N$ & Bias & MSE & Plug-In & Boot & Bias & MSE & Plug-In & Boot \\ 
\hline \\
& \multicolumn{4}{c}{p = 1} & \multicolumn{4}{c}{p = 3} \\ 
\hline \\
100 & 0.06 & 0.14 & 0.94 & 0.93 & 0.19 & 0.57 & 0.94 & 0.89 \\
200 & 0.04 & 0.1 & 0.93 & 0.94 & 0.07 & 0.21 & 0.96 & 0.92 \\
300 & 0.04 & 0.08 & 0.94 & 0.93 & 0.05 & 0.14 & 0.96 & 0.93 \\
500 & 0.03 & 0.07 & 0.92 & 0.92 & 0.04 & 0.1 & 0.94 & 0.91 \\
\hline \\
& \multicolumn{4}{c}{p = 5} & \multicolumn{4}{c}{p = 7} \\
\hline \\
100 & 17.81 & 116.47 & 0.92 & 0.73 & 1438.41 & 15608.04 & 0.56 & 0.33 \\
200 & 3.1 & 20.47 & 0.85 & 0.74 & 193.21 & 1830.28 & 0.49 & 0.38 \\
300 & 1.15 & 6.1 & 0.79 & 0.73 & 143.1 & 1470.73 & 0.44 & 0.33 \\
500 & 0.51 & 3 & 0.76 & 0.72 & 30.68 & 257.34 & 0.48 & 0.39 \\
\hline \\
\end{tabular}
\caption*{Notes.  Results are based on $1, 000$ simulation runs.  Bias is the difference between the true parameter value and the estimate constructed in Definition \ref{def:envscore} for Example \ref{ex:welfare}, averaged across simulation runs. MSE is mean squared error. Coverage Rate is the fraction of times a two-sided symmetric CI  with critical values $c_{\alpha/2}$ and $c_{1-\alpha/2}$ covers the true parameter, where $\alpha=0.95$ is the nominal coverage. A plug-in CI is given in \eqref{eq:plugin} and a weighted bootstrap CI is given in \eqref{eq:boot}, respectively, where $B=1000$ is the number of bootstrap repetitions.  $N$ is the sample size in each simulation run. }
\end{table}

\begin{table}[ht]
\centering
\caption{Finite-Sample Performance of Plug-In and Bootstrap CI (Gaussian Design)}
\label{tab:welfare2}
\begin{tabular}{ccccccccc}
 \hline \\
& & & \multicolumn{2}{c}{Coverage} & & & \multicolumn{2}{c}{Coverage} \\
\\
$N$ & Bias & MSE & Plug-In & Boot & Bias & MSE & Plug-In & Boot \\ 
\hline \\
& \multicolumn{4}{c}{p = 1} & \multicolumn{4}{c}{p = 3} \\ 
\hline \\
100 & 0.04 & 0.15 & 0.92 & 0.92 & 0.08 & 0.17 & 0.91 & 0.9 \\
200 & 0.02 & 0.1 & 0.95 & 0.95 & 0.03 & 0.11 & 0.94 & 0.94 \\
300 & 0.01 & 0.08 & 0.95 & 0.95 & 0.02 & 0.09 & 0.94 & 0.94 \\
500 & 0.01 & 0.06 & 0.95 & 0.94 & 0.01 & 0.07 & 0.94 & 0.94 \\
\hline 
& \multicolumn{4}{c}{p = 5} & \multicolumn{4}{c}{p = 7} \\
\hline \\
100 & 0.23 & 0.48 & 0.87 & 0.78 & 1.95 & 7.92 & 0.86 & 0.57 \\
200 & 0.07 & 0.15 & 0.92 & 0.9 & 0.32 & 1.11 & 0.9 & 0.78 \\
300 & 0.04 & 0.1 & 0.93 & 0.92 & 0.13 & 0.33 & 0.91 & 0.84 \\
500 & 0.02 & 0.07 & 0.94 & 0.94 & 0.05 & 0.13 & 0.94 & 0.91  \\
\hline \\
\end{tabular}
\caption*{Notes.  Results are based on $1, 000$ simulation runs.  Bias is the difference between the true parameter value and the estimate constructed in Definition \ref{def:envscore} for Example \ref{ex:welfare}, averaged across simulation runs. MSE is mean squared error. Coverage Rate is the fraction of times a two-sided symmetric CI  with critical values $c_{\alpha/2}$ and $c_{1-\alpha/2}$ covers the true parameter, where $\alpha=0.95$ is the nominal coverage. A plug-in CI is given in \eqref{eq:plugin} and a weighted bootstrap CI is given in \eqref{eq:boot}, respectively, where $B=1000$ is the number of bootstrap repetitions.  $N$ is the sample size in each simulation run. }
\end{table}

\subsection{Empirical application}
\label{sec:empirical}

To illustrate the immediate applicability of the proposed method, this study analyzes the effect of Medicaid exposure on healthcare utilization and health outcomes using data from the Oregon Health Insurance Experiment \citep{finkelstein}. In 2008, Oregon implemented a limited expansion of its Medicaid program, providing insurance coverage to low-income, uninsured adults selected through a lottery system from a waiting list. One year after randomization, a subset of $N = 58,405$ applicants was mailed a survey to assess changes in healthcare utilization and general well-being, with a response rate of approximately $50\%$. Abstracting from potential non-response bias, the study found that Medicaid significantly improved healthcare access, financial security, and mental health outcomes for low-income adults.\footnote{\citep{finkelstein} reported that the ability to reject the null hypothesis of no effect of health insurance on healthcare utilization or financial strain is generally robust to Lee bounds, while the ability to reject the null hypothesis of no effect on self-reported health outcomes is not robust (see footnote 19).} To examine the robustness of these findings, this section reports various versions of Horowitz-Manski-Lee bounds under various assumptions about subjects' response behavior.

This section focuses on the average treatment effect (ATE) of Medicaid exposure on self-reported mental health. Let $S(1) = 1$ and $S(0) = 1$ be binary indicators denoting whether a subject completes a survey when treated or not treated, respectively, and let $Y(1)$ and $Y(0)$ represent the corresponding potential outcomes. The observed data, $W = (X, D, S, S \cdot Y)$, include $D$ (lottery outcome), $S$ (an indicator of non-missing response), $Y$ (the response itself),  and baseline covariates $X$. The propensity score, $\mu(X)$, is determined by household size and survey wave fixed effects.\footnote{If an applicant wins the lottery, all members of their household become eligible to enroll. Consequently, larger households have a higher probability of winning the lottery compared to smaller ones. Additionally, control applicants were oversampled in earlier survey waves, further influencing the propensity score.} Other components of $X$ include 64 predetermined characteristics, such as demographics, enrollment in the Supplemental Nutrition Assistance Program (SNAP) or Temporary Assistance for Needy Families (TANF) as well as the total amount of benefits received in each program, and pre-existing health conditions.

Table \ref{tab:JC} summarizes the findings. The baseline estimate of Medicaid's effect on mental health is $2.27\%$. Interestingly, the control group's response rate ($49.4\%$) exceeds that of the treated group ($48.2\%$). Unconditional monotonicity (Assumption \ref{ass:identification:treat}), which posits that treatment discourages survey completion, lacks an intuitive explanation. Standard Lee bounds (Column (1)) are misleadingly tight.  Without monotonicity, the proportion of ``always-takers'' (respondents regardless of treatment status) cannot be bounded away from zero, and no-monotonicity bounds (Column (5)) do not provide any meaningful restriction on the treatment effect\begin{footnote}{\citep{ZhangRubin} bounds reported in Column (5) are driven by $Y(1), Y(0) \in \{1,0\}$.}\end{footnote}. This pattern holds for all survey outcomes, including questions about healthcare  utilization, financial strain, and self-reported health outcomes.

To make progress, the analysis focuses on a subset of the population for whom the direction of selection response aligns with the majority. Specifically, the parameter of interest is: \begin{align} \label{eq:bb0} \beta_0 = \E [ Y(1) - Y(0) \mid S(1) = S(0), s_0(0, X) \geq s_0(1, X)] \end{align} where the outcome of interest is mental health as measured by a positive response to the question “Did you not screen positive for depression in the last two weeks?”). The selection equation \eqref{eq:sel1} is estimated using logistic regression (see Example \ref{ex:sel}). The estimated share of subjects with negative selection response is $82\%$ (Column (2)), consistent with the negative direction of unconditional response effect. A higher likelihood of survey response is associated with being female, requesting English-language materials, and not receiving TANF benefits. Additionally, Medicaid's effect on response appears negative for individuals who experienced injuries or received SNAP benefits prior to randomization, suggesting that control participants' response could be driven by  acute health or financial challenges. 

Basic bounds on \eqref{eq:bb0} (Table \ref{tab:JC}, Column (3)), derived under conditional monotonicity,  cannot determine the direction of the treatment effect. Sharp bounds (Table \ref{tab:JC}, Column (4)), constructed using Algorithm \ref{alg:lee} with first-stage outcome fitted values as described in Examples \ref{ex:sel} and \ref{ex:outcome}, suggest that the Medicaid exposure effect is positive, though the magnitude is attenuated at the lower bound. The proposed approach relies on a smoothness assumption, specifically that the conditional selection and outcome probabilities are sufficiently continuously distributed, which could be plausible  since  some components of $X$ are continuously distributed, such as total amount of SNAP or TANF benefits or pre-existing ED charges.

\begin{table}
   \captionof{table}{Bounds on the Medicaid effect on self-reported mental health }
   \label{tab:JC}
   \centering 

   \begin{tabular}{lcccccc}
        \toprule
     &  \shortstack{Unconditional \\ Monotonicity}  &      \multicolumn{3}{c}{Conditional  Monotonicity} & No Monotonicity  \\ 
    \\
   & (1) & (2) & (3) & (4) & (5) \\
     & & \shortstack{$\Pr (\tau(X) \geq 0)$} & Basic & Sharp & Basic \\ 
    \midrule
 
 Bounds       &           [0.014,             0.039]              &  0.827   &         [-0.003,  0.053] &  [0.002, 0.054] & [-1, 1]  \\
        95\% CR  &    (-0.004, 0.053) &  & (-0.026,   0.070)  &  (-0.014,   0.072)  & (-1, 1)    \\
        \bottomrule
    \end{tabular}

   \caption*{
   Notes. 
    Table shows estimated  bounds  in square brackets and  the $95 \%$ confidence region for the identified set  in parentheses. Column (1): basic bounds  on regular ATE  $\E[ Y(1) - Y(0) \mid S(1) = S(0)=1]$ under Assumption \ref{ass:identification:treat} with $ S(0) \geq S(1) \text{ a.s.}$. Columns (2)--(4)  report the results assuming conditional monotonicity. Column (2): estimated share of subjects with non-positive selection response. Column (3): basic bounds on the parameter \eqref{eq:bb0}. Column (4): bounds in  Algorithm \ref{alg:lee} where selection equation estimated in Example \ref{ex:sel} and the outcome equation as in Example \ref{ex:outcome}. Column (5): basic \citep{ZhangRubin} no-monotonicity bounds on the parameter \eqref{eq:bb0}.     Computations use design weights. The sample size $N = 58, 405$.    The asymptotic probability of the 95\% Confidence Region is based on $B = 500$ bootstrap repetitions.  }
\end{table}

\appendix 
\newpage

\renewcommand{\theequation}{A.\arabic{equation}}
\renewcommand{\thelemma}{A.\arabic{lemma}}
\renewcommand{\theassumption}{A.\arabic{assumption}}
\renewcommand{\theremark}{A.\arabic{remark}}
\renewcommand{\thecorollary}{A.\arabic{corollary}}
\renewcommand{\thetable}{A.\arabic{table}}
\renewcommand{\thesection}{A} 

\setcounter{remark}{0}

\setcounter{assumption}{0}
\section{Proofs}
\label{sec:proofs}

\paragraph{Notation.}  I use the empirical process notation. For a generic function $f$ and a generic sample $(W_i)_{i=1}^N$, denote the empirical sample average by $$ \EN f(W_i) := \dfrac{1}{N}  \sum_{i=1}^N f(W_i) $$ and the
scaled, demeaned sample average by  $$\GN f(W_i) := 1/\sqrt{N} \sum_{i=1}^N  [f(W_i) - \int f(w) d P(w)].$$ 
For two sequences of random variables  $\{ a_N, b_N, N \geq 1\}: a_N \lesssim_{P}  b_N$ means    $a_N = O_{P} (b_N)$. For two sequences of numbers $\{a_N, b_N, N \geq 1\}$, $a_N \lesssim  b_N$ means $a_N = O (b_N)$. Let $a \wedge b = \min \{ a, b\}, a \vee b = \max \{ a, b\} $. The $\ell_2$ norm of a vector is denoted by $\| \cdot \|$, the $\ell_1$ norm is denoted by $\| \cdot \|_1$, the $\ell_{\infty}$ norm is denoted by $\| \cdot \|_{\infty}$, and $\ell_0$ norm is denoted by $\| \cdot \|_{0}$. The notation $\| f \|_{F,2}$ stands for $(\int_{\mathcal{W}} f^2(w)  F(dw))^{1/2} = \| f \|_{F,2}$.

\paragraph{Auxiliary Statements.} Let $\Zeta$ be a random variable whose CDF and quantile function are denoted by $F_{\Zeta}$ and $Q_{\Zeta}$, respectively. Given a quantile level $\omega \in (0,1)$,  define lower- and upper-truncated random variables
\begin{align}
F^{L}_{\Zeta}(t) = \begin{cases}
\dfrac{F_{\Zeta}(t) - \omega}{\omega}, &\quad t \leq Q_{\Zeta}(\omega), \\
0, &\quad t > Q_{\Zeta}(\omega) \\
\end{cases}, \quad F^{U}_{\Zeta}(t) = \begin{cases}
0, &\quad  t< Q_{\Zeta}(1-\omega) \\
\dfrac{F_{\Zeta}(t) - \omega}{\omega}, &\quad t \geq Q_{\Zeta}(1-\omega).
\end{cases}
\end{align}
Define the lower- and upper- conditional value-at-risk \citep{RockUryasev} as
\begin{align}
CVAR^{L}_{\Zeta} (\omega) =\int_{-\infty}^{\infty} t d F^{L}_{\Zeta}(t), \qquad CVAR^{U}_{\Zeta} (\omega) =\int_{-\infty}^{\infty} t d F^{U}_{\Zeta}(t).
\end{align}
If $\Zeta$ has a continuous distribution without point masses,   $$F_{\Zeta} (Q_{\Zeta}(\omega)) = \omega, \quad CVAR^{L}_{\Zeta} (\omega) =   \E[ Z \mid Z \leq Q_Z(\omega)], \quad CVAR^{U}_{\Zeta} (1-\omega) = \E[ Z \mid Z \geq Q_Z(1-\omega)].$$ Finally,  $CVAR^L_{\Zeta} (\omega)$ and $CVAR^U_{\Zeta} (\omega)$ can be equivalently expressed as intersection bounds 
\begin{align}
&CVAR^{L}_{\Zeta} (\omega)  = \sup_{\beta \in \mathrm{R}}  \left( \beta + \omega^{-1}   \E_{\Zeta}[ \min (Z-\beta,0)]  \right)  \label{eq:cvarl} \\
&CVAR^{U}_{\Zeta} (\omega)  = \inf_{\beta \in \mathrm{R}}  \left( \beta + \omega^{-1}   \E_{\Zeta}[ \max (Z-\beta,0)]  \right). \label{eq:cvaru}
\end{align}
In case of \eqref{eq:cvarl}, the sup is attained by the value of $\beta$ that equals $\omega$-quantile
$$
Q_Z(\omega)= \inf \{\beta: F_Z(\beta) \geq \omega \}.
$$

\begin{lemma}[\citet{LuedtkeLaan}, Theorem 1]
\label{lem:LvdL}
Suppose the outcome $Y$ is a.s. bounded by $M$. Then, the first-best welfare  $\E [\max (m(1,X), m(0,X))]$ is pathwise differentiable if and only if 
\begin{align}
\label{ass:unique}
P_X (m(1,X) = m(0,X))=0
\end{align}
holds. The efficient score is
\begin{align}
\label{eq:efficient}
g(W, \xi_0) &= (\rho(W, 1, \xi_0) - \rho(W, 0, \xi_0))  1\{ m(1,X) - m(0,X)>0 \} +  \rho(W, 0, \xi_0).
\end{align}
\end{lemma}

Lemma \ref{lem:LvdL} restates Theorem 1 in  \citet{LuedtkeLaan}.  If  \eqref{ass:unique} holds, the sign of $m(1,X) - m(0,X)$ can be treated as known when calculating the efficient score. 

\subsection{Proofs for Section \ref{sec:res}}

\begin{proof}[Proof of Theorem \ref{thm:closedform}]

The first step introduces error terms. The second and third steps  bound the first-order and the second-order terms.

\textbf{ Step 1. } Define the true and estimated minimizers
\begin{align*}
t_0(X) := \arg \min_{t \in \mathcal{T}} \phi (t,\nu_0(X)), \quad t(X) :=\arg \min_{t \in \mathcal{T}}  \phi(t,\nu(X)).
\end{align*}
Define the errors  due to mistakes in estimated minimizers
\begin{align*}
 \tau_0(X) :&= \phi (t(X), \nu_0(X))- \phi (t_0(X), \nu_0(X)), \\
   \tau(X) :&= \phi(t(X),\nu(X)) - \phi (t_0(X), \nu(X)).
\end{align*}
Finally, decompose the estimation error of the proposed signal
\begin{align}
 \rho (W,t,\xi) - \rho(W,t_0,\xi_0) = (\rho (W,t,\xi) - \rho(W,t,\xi_0)) + (\rho(W,t,\xi_0) - \rho(W,t_0,\xi_0)) \nonumber  \\
 &=: S_1 + S_2. \label{eq:s1s2}
\end{align}
 
\textbf{ Step 2. }  By definition of $t_0(X)$, the error term $\tau_0(X)$ must be non-negative with probability one
$$
\Pr (\tau_0(X) \geq 0) = 1.
$$
Furthermore, since the true minimizer is assumed unique, $\tau_0(x)=0$ if and only if $t(x) = t_0(x)$ for any $x \in \mathcal{X}$. As a result, the probability of $\tau_0(X)$ being in an open interval  $(0,t)$ is   upper bounded by Assumption \ref{ass:ma}
\begin{align*}
\Pr ( 0 <  \tau_0(X) < t) \leq \Pr \left( 0 \leq \min_{t \in \mathcal{T} \setminus \arg \min_{t \in \mathcal{T}} \phi (t, \nu_0(X)) } \phi (t, \nu_0(X))- \min_{t \in \mathcal{T}} \phi (t, \nu_0(X)) \leq t \right) \leq  \bar{B} t.
\end{align*}

\textbf{ Step 3. }    Next, by definition of $t(X)$ in Step 1,  $\tau(X) \leq 0$ with probability one. Then 
\begin{align}
\label{eq:etau}
\tau(X) \leq 0  <  \tau_0(X) \Rightarrow 0 < \tau_0(X) \leq  \tau_0(X) -\tau(X).
\end{align}
For any element $\nu$ in the shrinking neighborhood $\mathcal{T}^{\nu}_N$ and any $x \in \mathcal{X}$,  the upper bound applies
\begin{align}
|\tau(x) - \tau_0(x) | &\leq | \phi(t(x),\nu(x)) - \phi (t(x), \nu_0(x)) | \nonumber \\
&+ | \phi(t_0(x),\nu(x)) - \phi (t_0(x), \nu_0(x)) | \leq^{(i)} 2 B_{\phi} \| \nu(x) - \nu_0 (x) \| \leq^{(ii)} 2 B_{\phi} \nu^{\infty}_N, \label{eq:uniformrate}
\end{align}
where (i) follows from Assumption \ref{ass:reg}(2) and (ii) from Assumption \ref{ass:rate}. Define  the misclassification event
\begin{align*}
 \mathcal{E}_{\tau}  := \{ 0 < \tau_0(X) \leq 2 B_{\phi} \nu^{\infty}_N \}.
\end{align*}
Invoking Assumption \ref{ass:smallbias} gives an upper bound for the first term 
\begin{align*}
\sqrt{N} | \E [ S_1 ] | \leq \sqrt{N} \sup_{t \in \mathcal{T}} \sup_{\xi \in \Xi_N} | \E [ \rho (W,t,\xi) - \rho(W,t,\xi_0)] | \leq B_N = o(1).
\end{align*}
To bound the second term, I show that 
\begin{align*}
| \E [ S_2 ] |   = O( (\nu^{\infty}_N)^2).
\end{align*}
By LIE,  the second term can be upper bounded 
\begin{align*}
0 \leq \E [ S_2 ]  =  \E [ \tau_0(X)  ] &\leq  \E [ (\tau_0(X)- \tau(X))  1\{  0 < \tau_0(X) \leq  \tau_0(X) -\tau(X) \} ].
\end{align*}
For any $\nu \in \mathcal{T}^{\nu}_N$, 
$$
\{ X: 0 < \tau_0(X) \leq  \tau_0(X) -\tau(X) \} \Rightarrow \{ X \in  \mathcal{E}_{\tau} \},
$$
and 
\begin{align}
\label{eq:maineq0}
\E [ (\tau_0(X)- \tau(X))  1\{  0 < \tau_0(X) \leq  \tau_0(X) -\tau(X) \}]  \leq \E [ (\tau_0(X)- \tau(X))  1\{ \mathcal{E}_{\tau} \}] \leq  2 B_{\phi} \nu^{\infty}_N \Pr ( \mathcal{E}_{\tau}  ),
\end{align}
which, in turn, is upper bounded as 
\begin{align}
\label{eq:maineq}
\Pr ( \mathcal{E}_{\tau}  ) =   \Pr (   0 < \tau_0(X)  \leq  2 B_{\phi} \nu^{\infty}_N )  \leq  2 B_{\phi} \bar{B} \nu^{\infty}_N.
\end{align}
Combining the displays \eqref{eq:maineq0} and \eqref{eq:maineq} gives
\begin{align*}
| \E [ S_2 ] |   &\leq  2 B_{\phi} \nu^{\infty}_N \Pr ( \mathcal{E}_{\tau}  )  \leq  4 B^2_{\phi} \bar{B} (\nu^{\infty}_N)^2.
\end{align*}
Adding the bounds  on $\E [ S_1 ]$ and $\E [ S_2 ]$ gives
\begin{equation}
 \sup_{\xi \in \Xi_N} \sqrt{N} |\E [ S_1 + S_2]|= O(  B_N +   \sqrt{N}  (\nu^{\infty}_N)^2) = o(1). 
\end{equation}
Finally, the second-order terms can be bounded as 
\begin{align*}
 \sup_{\xi \in \Xi_N} \E [ ( \rho(W,t,\xi) - \rho(W, t_0,\xi_0))^2] &\leq 2  \sup_{\xi \in \Xi_N}  (\E [ S^2_1  ] + \E [ S^2_2 ])  = O(\Lambda_N + \nu^{\infty}_N) = o(1),
\end{align*}
where the bound on $\E [ S^2_1  ]$ follows from Assumption \ref{ass:smallbias}. To see the bound on $\E [ S^2_2]$  note 
$$
 \sup_{\xi \in \Xi_N} \E [ S^2_2 ] \leq  \sup_{x \in \mathcal{X}} \sup_{t \in \mathcal{T}} 2  \E [ \rho^2(W, t, \xi_0) \mid X=x]  \sup_{\nu \in \mathcal{T}^{\nu}_N}\Pr ( \mathcal{E}_{\tau}  )  = O ( \nu^{\infty}_N) 
$$
Invoking Lemma A.3 from \citep{SemCher} gives the statement of the Theorem.
\end{proof}

\begin{proof}[Proof of Theorem \ref{thm:consistency}]
Decompose the estimation error of the proposed moment function
\begin{align*}
 \rho^2 (W, t, \xi) -  \rho^2 (W, t_0, \xi_0)&= ( \rho^2 (W, t, \xi) - \rho^2 (W, t, \xi_0)) + (\rho^2 (W, t, \xi_0) -  \rho^2 (W, t_0, \xi_0)) \\
 &= T_1 + T_2.
\end{align*}
For any a.s. $B$-bounded random variables $P$ and $Q$, note that 
\begin{align}
\label{eq:mainineq2}
\E [ |P^2 - Q^2| ] \leq 2 \E [|P - Q| \max (|P|, |Q|)] \leq 2 B \| P-Q \|_{P,2}.
\end{align}
Plugging $P:= \rho (W, t, \xi)$ and $Q:=  \rho(W, t, \xi_0)$ and $B=B_{\rho}$ gives an upper bound on the first term
$$
\sup_{ \xi \in \Xi_N} | \E [ T_1] | \leq^{i} 2B_{\rho}  \sup_{t \in \mathcal{T}} \sup_{ \xi \in \Xi_N}  (\E [( \rho (W, t, \xi) -  \rho(W, t, \xi_0))^2])^{1/2}  \leq^{ii} 2 B_{\rho} \Lambda_N^{1/2}
$$
where (i) follows from \eqref{eq:mainineq2} as well as an upper bound in Assumption \ref{ass:reg} and (ii) from Assumption \ref{ass:smallbias}. To bound the second term, 
$$
\sup_{ \xi \in \Xi_N}  | \E [T_2]  | \leq^{i} 2B_{\rho}^2 \sup_{ \nu \in \mathcal{T}^{\nu}_N}  \Pr ( t(X) \neq t_0(X)) \leq 2 B_{\rho}^2 \sup_{ \nu \in \mathcal{T}^{\nu}_N}  \Pr ( \mathcal{E}_{\tau}  )  \leq^{ii} 4\bar{B} B_{\phi} B_{\rho}^2 \nu^{\infty}_N.
$$
where (i) follows from an upper bound in Assumption \ref{ass:reg} and (ii) from \eqref{eq:maineq}. Combining the terms gives a bound on bias $|\E [ T_1 + T_2 ] | = O(\Lambda_N^{1/2} + \nu^{\infty}_N) = o(1)$. 
The consistency follows from a standard LLN invoked for each partition $k \in [K]$, where the summands are i.i.d. conditional on the data in the hold-out partitions $(W_i)_{i \in J_k^c}$. 
\end{proof}

\begin{proof}[Proof of Theorem \ref{thm:boot}]
Theorem \ref{thm:boot} is a special case of Corollary 3.2 in \citep{SemJoE} which itself follows from Theorem 3.2 therein. Assumption 3.1 holds with $\Sigma=1 = A(W, \eta_0) \text{ a.s.} $. Since the problem is scalar (i.e. $d=1$), Assumption 3.2 does not apply. Next, it suffices to verify Assumption 3.3 with $\mu_N $ and $r_N'$ and  $A_N=\delta_N= r_N''=0$. From the proof of Theorem \ref{thm:closedform}, 
\begin{align*}
\mu_N = O (B_N +     N^{1/2} (\nu^{\infty}_N)^2) = o (1), \quad r_N' = O(\Lambda_N + \nu^{\infty}_N) = o(1).
\end{align*}
Assumptions 3.4 and 3.5 do not apply. The result follows.
\end{proof}

\begin{proof}[Proof of Lemma \ref{ex:1}]
The proof is established in two steps. In Step 1, I assume that $g_t(\cdot)=0$ for all $t \in \mathcal{T}$ to establish the result by verifying Assumption \ref{ass:ma}.  In Step 2, I drop this assumption.  

\textbf{ Step 1 }. The following inequality holds
\begin{align*}
&    \sup_{ (t,j) \in \mathcal{T}: t \neq j}  \Pr ( 0 < | X^{\prime} (\gamma_t - \gamma_j) | < t)  \leq \sup_{ (t,j) \in \mathcal{T}: t \neq j}  \Pr ( 0 < | X^{\prime} (\gamma_t - \gamma_j) | < (t/\| \gamma_t - \gamma_j \|) \| \gamma_t - \gamma_j \|)  \\
&\leq^{i} \sup_{ (t,j) \in \mathcal{T}: t \neq j} (B_X/ \| \gamma_t - \gamma_j \|) t,
\end{align*}
where (i) follows from \eqref{eq:macov}. Since $\gamma_t \neq \gamma_j$ for all $k \neq j$, Assumption \ref{ass:ma} holds with a finite constant $\bar{B} = \sup_{ (t,j) \in \mathcal{T}: t \neq j} (B_X/ \| \gamma_t - \gamma_j \|)$. 

\textbf{ Step 2 }. For each $t, j$, the conditional distribution of  $m(t,\widetilde{X} )-m(j,\widetilde{X} ) \mid  \bar{X}$ is the same as the unconditional distribution of $X^{\prime} (\gamma_t - \gamma_j)$ up to a finite shift given by 
$g_t(\bar{X}) - g_j (\bar{X})$. Thus, since  $X^{\prime} (\gamma_t - \gamma_j)$ has a $\bar{B}$-bounded unconditional density,  the random variable $X^{\prime} (\gamma_t - \gamma_j) + g_t(\bar{X}) - g_j (\bar{X})$ has a bounded conditional density that is $\bar{B}$-bounded  a.s. in $\bar{X}$ and thus $\bar{B}$-bounded unconditionally.

\end{proof}

\begin{proof}[Proof of Lemma \ref{ex:3}]
The following inequality holds
\begin{align*}
 \sup_{  (t,j) \in \mathcal{T}: t \neq j}  \Pr ( 0 < | m(t,X) - m(j,X) | < t) &=   \sup_{  (t,j) \in \mathcal{T}: t \neq j}  \Pr ( 0 < | F(X^{\prime} \gamma_t  ) - F(X^{\prime} \gamma_j)    | < t) \\
&\leq   \sup_{  (t,j) \in \mathcal{T}: t \neq j}  \Pr ( 0 < \underline{f}  | X^{\prime} (\gamma_t-\gamma_j)   | < t  ) \\
&\leq^{i} \sup_{ (t,j) \in \mathcal{T}: t \neq j} (B_X/  \underline{f} \| \gamma_t - \gamma_j \|)  t.
\end{align*}
where (i) follows from \eqref{eq:macov} and the conditions of Lemma \ref{ex:3}. 
\end{proof}

\subsection{Proofs for Section \ref{sec:appl}}

\begin{proof}[Proof of Proposition \ref{prop:roy2}]
Invoking the proof of Proposition 1 in \citet*{MourifieHenry} conditional on $X=x$ for each value of $x \in \mathcal{X}$ gives 
 \begin{align}
\Pr (Y(1)=1, Y(0) =0 \mid X=x) &\leq \min_{z \in \mathcal{Z}} \Pr (Y=1, D=1\mid Z=z, X=x) \label{eq:roy2app} \\
\Pr (Y(1)=0, Y(0) =1 \mid X=x) &\leq  \min_{z \in \mathcal{Z}}  \Pr (Y=1, D=0 \mid Z=z, X=x) \label{eq:roy3app}   
\end{align}
Averaging over marginal covariate distribution and invoking LIE gives
 \begin{align}
\Pr (Y(1)=1, Y(0) =0) \leq \E  [\min_{z \in \mathcal{Z}} \Pr (Y=1, D=1\mid Z=z, X)] \label{eq:roy22app} \\
\Pr (Y(1)=0, Y(0) =1) \leq \E [\min_{z \in \mathcal{Z}}  \Pr (Y=1, D=0 \mid Z=z, X)]\label{eq:roy33app} 
\end{align}
Furthermore, under unconditional independence \eqref{eq:uncondindep}, the propensity score is constant and 
\begin{align*}
f_{X \mid Z=z} (x \mid Z=z) = f_X (x),
\end{align*}
which implies
\begin{align*}
\E_X [ \Pr (Y=1, D=1\mid Z=z, X) \mid Z=z ] =  \Pr (Y=1, D=1\mid Z=z).
\end{align*}
As a result, invoking Jensen's inequality gives
 \begin{align*}
\E [\min_{z \in \mathcal{Z}} \Pr (Y=1, D=1\mid Z=z, X)]  \leq \min_{z \in \mathcal{Z}} \Pr (Y=1, D=1\mid Z=z) \\
\E [\min_{z \in \mathcal{Z}} \Pr (Y=1, D=0\mid Z=z, X)]  \leq \min_{z \in \mathcal{Z}} \Pr (Y=1, D=0\mid Z=z).
\end{align*}
\end{proof}

\begin{proof}[Proof of Corollary  \ref{cor:roy}]
Assumption \ref{ass:smallbias} follows from Assumption \ref{ass:fsrate}, as shown in Lemma 4.11 in \cite{SemCher}. Assumption \ref{ass:rate} is directly assumed.  Since $D$ and $Y$ are binary random variables, Assumption \ref{ass:reg} (1) holds with an upper bound of
$$
\sup_{\nu \in \mathcal{T}^{\nu}_N} \sup_{w \in \mathcal{W}} \sup_{z \in \mathcal{Z}} | \rho (w, z, \nu) | \leq 1 + 2/\kappa  \text{ a.s.},
$$
which implies Assumption \ref{ass:reg}(1) is satisfied. Assumption \ref{ass:reg} (2) holds with $B_{\phi} = 1$. Assumption \ref{ass:ma} is assumed directly. \end{proof}

\begin{proof}[Proof of Proposition \ref{prop:lee}]
Proposition \ref{prop:lee} is a special case of Proposition \ref{prop:leed} with $\mathrm{T} = \{ 0, 1 \}$. For a binary outcome $Y$ taking values $1$ and $0$, 
\begin{equation}
\E [ \max(Y-\beta, 0) \mid D=1, S=1, X=x] = p_Y (x) \max (1-\beta, 0) + \max (-\beta, 0) (1-p_Y(x)). \label{eq:maxbeta} 
\end{equation}
Then,
\begin{align*}
&\inf_{\beta \in \mathrm{R} } (\beta s_0(0,x) +   s_0(1,x )  \E [ \max(Y-\beta, 0) \mid X=x]) \\
  &= s_0(0,x) + \inf_{\beta \in \mathrm{R} } \{  (\beta-1) s_0(0,x) +   s_0(1,x )  (p_Y (x) \max (1-\beta, 0) + \max (-\beta, 0) (1-p_Y(x))) \} \\
&= s_0(0,x) + \min (0, s_0(1,x )  p_Y(x)-s_0(0,x)).
\end{align*}
Consider the case when $\beta \leq 0$. Then, 
\begin{align*}
&\inf_{\beta \leq 0}  (1-\beta) (s_0(1,x )  p_Y(x) - s_0(0,x)) +  s_0(1,x )  \max (-\beta, 0) (1-p_Y(x)))  \\
&= \inf_{\beta \leq 0}  \beta (s_0(0,x) - s_0(1,x)) + s_0(1,x )  p_Y(x) - s_0(0,x)  \\
&\leq  s_0(1,x )  p_Y(x) - s_0(0,x)
\end{align*}
since $s_0(0,x) - s_0(1,x) \leq 0 \quad \forall x \in \mathcal{X}$ by Assumption \ref{ass:identification:treat}. For $\beta \geq 0$, 
\begin{align*}
&\inf_{\beta \geq 0}  (1-\beta) (s_0(1,x )  p_Y(x) - s_0(0,x)) \geq \min (0, s_0(1,x )  p_Y(x)-s_0(0,x)).
\end{align*}
Combining the results gives an expression that coincides with the numerator $N_U$ in \eqref{eq:nubinary}. A similar argument for the numerator $N_L$ applies.
\end{proof}

\begin{proof}[Proof of Proposition \ref{prop:leed}]
Step 1 establishes a version of \eqref{eq:nl} and \eqref{eq:nu}  with the index set  $\mathcal{T} = \mathrm{R}$. Step 2 shows that the index set $\mathcal{T}$ can be reduced to the support set $\mathrm{T}$.

\textbf{Step 1}. As shown in \citep{HorowitzManski}, the distribution of treated outcomes  is a mixture of always-takers' and compliers' outcomes with mixing proportions $p_0(x)$ and $1-p_0(x)$ where $p_0(x)$ is given in \eqref{eq:condtrim}.  Therefore, the conditional always-takers' ATE $\beta_1(x)= \E[Y(1)  \mid S(1)=1, S(0)=1, X=x]$ is bounded from above and below as
\begin{align}
CVAR^{L}_{Y \mid D=1, S=1, X=x} (p_0(x)) \leq \beta_1(x) \leq CVAR^{U}_{Y \mid D=1, S=1, X=x} (p_0(x)).
\end{align}
The always-takers' covariate density is obtained by Bayes rule
$$
f_X (x \mid S(1) = S(0)=1) = \dfrac{s_0(0,x) f_X(x)}{\E [ s_0(0,X)]}.
$$
Multiplying by $f_X (x \mid S(1) = S(0)=1)$ and averaging gives 
$$
\dfrac{\E [ CVAR^{L}_{Y \mid D=1, S=1, X} (p_0(X)) s_0(0,X) ]  }{\E [s_0(0,X)]} \leq \beta_1 \leq \dfrac{\E [ CVAR^{U}_{Y \mid D=1, S=1, X} (p_0(X)) s_0(0,X) ]  }{\E [s_0(0,X)]},
$$
Replacing $ CVAR^{L}_{Y \mid D=1, S=1, X=x} (p_0(x)) $ and  $ CVAR^{U}_{Y \mid D=1, S=1, X=x} (p_0(x)) $ with their dual analogs  \eqref{eq:cvarl} and \eqref{eq:cvaru}  gives an expression for  $N_L$ and $N_U$ that are identical to those  in \eqref{eq:nl} and \eqref{eq:nu} except $\mathrm{T} = \mathrm{R}$.

\textbf{Step 2}. Consider \eqref{eq:nu}. Let $\{ y_k \}_{k=1}^{| \mathrm{T}|}$ be the ordered set of support points where $y_{\max} = y_{| \mathrm{T}|}$.  Note that 
\begin{align*}
\E [ \min (Y -\beta, 0) \mid D=1, S=1, X=x] &= \sum_{y \leq \beta, y \in \mathrm{T}  } (y-\beta) \pi_{y0}(x), \\
\E [ \max (Y -\beta, 0) \mid D=1, S=1, X=x] &= \sum_{y \geq \beta, y \in \mathrm{T}  } (y-\beta) \pi_{y0}(x).
\end{align*}
Note that $\mathcal{T}$ can be taken as   $[ y_1, y_{\max} ]$ since $[y_{\max}, \infty)$ and $(-\infty, y_{\min}]$ are dominated by an argument similar to the proof of Proposition \ref{prop:leed}.
  By Assumption \ref{ass:identification:treat}, $(\beta'-\beta) (s_0(0,x) -s_0(1,x))<0$ for any $\beta<\beta' $ where $\beta, \beta' \in (y_k, y_{k+1}]$. Thus, for each $\beta \in (y_{k}, y_{k+1}]$ the bound expression is minimized at the right endpoint of the interval, namely $y_{k+1}$.  Thus, it suffices to consider support points only as well as $(-\infty, y_1]$ and $(y_{k+1}, +\infty)$. For $\beta \in (y_{k+1}, +\infty)$, those are dominated by $y_{k+1}$. Likewise, for the points in $(-\infty, y_1)$ are dominated by $\beta=0$ since $\beta (s_0(0,x) - s_0(1,x)) \leq 0$.
\end{proof}

\begin{proof}[Proof of Corollary  \ref{prop:leed2}]
The numerator upper bound $N_U$ in \eqref{eq:nu} is a special case of  \eqref{eq:mainpsi} with
$$
\nu_0(x) =  (s_0(0,x), s_0(1,x), {\pi_{\beta0}(x)}_{\beta  \in \mathrm{T} \setminus \{ y_{\max} \} })
$$
and 
$$
\phi (\beta, \boldsymbol{v}) = \beta \boldsymbol{v}_1 + \boldsymbol{v}_2 \sum_{y \in \mathrm{T} \setminus \{ y_{\max} \} } 1\{ y \geq \beta \} (y-\beta) \boldsymbol{v}_{y} 
$$
Step 1 verifies Assumptions \ref{ass:smallbias} and \ref{ass:rate}. Step 2 verifies Assumption \ref{ass:reg}. Steps 3 verifies Assumption \ref{ass:ma}.

\textbf{ Step 1. }  Assumption \ref{ass:smallbias} holds by construction since the moment functions $\{\rho_U(W, \beta)\}_{\{\beta \in \mathrm{T}\}}$ in \eqref{eq:guhelp_beta} do not involve any nuisance parameters. Assumption \ref{ass:rate} is directly assumed in Assumption \ref{ass:fsrate2}.

\textbf{ Step 2. } Consider Assumption \ref{ass:reg} (1) holds with $B_{\rho} =2 \max_{y \in \mathrm{T}} |y|/\kappa$.  Next, note that 
\begin{align*}
| \dfrac{\partial \phi (x, \boldsymbol{v}) }{\partial \boldsymbol{v}_1} | &= |\beta | \leq \max_{y \in \mathrm{T}} |y|   \\
|  \dfrac{\partial \phi (x, \boldsymbol{v}) }{\partial \boldsymbol{v}_2} |  &\leq 2 \max_{y \in \mathrm{T}} |y| |  \mathrm{T} |   |\boldsymbol{v}_y| \leq 2 \max_{y \in \mathrm{T}} |y| |  \mathrm{T} | (1-\kappa) \\
  |\dfrac{\partial \phi (x, \boldsymbol{v})}{\partial \boldsymbol{v}_j}| &\leq 2 \max_{y \in \mathrm{T}} |y| |  \mathrm{T} |    |\boldsymbol{v}_2| \leq 2 \max_{y \in \mathrm{T}} |y| |  \mathrm{T} | (1-\kappa),
\end{align*}
which implies Assumption \ref{ass:reg} (2) holds with some $B_{\phi}$ large enough.

\textbf{ Step 3. } Let $\beta$ and $\beta'$ such that $\beta < \beta'$ be two distinct values in $\mathrm{T}$.  Then, 
\begin{align*}
&\phi (\beta', \nu_0(X)) - \phi (\beta, \nu_0(X)) \\
= (\beta' - \beta) s_0(0,X) &+ s_0(1,X) \sum_{y \in \mathrm{T}  } 1\{ y \geq \beta'  \}  ( y - \beta' ) \pi_{y0}(X)   \\
&- s_0(1,X) \sum_{y \in \mathrm{T}  } 1\{ y \geq \beta'  \}  ( y - \beta ) \pi_{y0}(X)  \\
&- s_0(1,X)  \sum_{y \in \mathrm{T} } 1\{  \beta \leq y <  \beta'  \}  ( y - \beta) \pi_{y0}(X)   \\
&= (\beta' - \beta) (s_0(0,X) - s_0(1,X) \pi_{y0}(X) \sum_{y \in \mathrm{T}  } 1\{ y \geq \beta'  \}  ( y - \beta )) \\
& - s_0(1,X)  \sum_{y \in \mathrm{T} } 1\{  \beta \leq y <  \beta'  \}  ( y - \beta) \pi_{y0}(X) 
\end{align*}
is a smooth mapping of the vector $(s_0(0,X), s_0(1,X), \{ \pi_{y0}(X) \}_{ y \in [\beta, \beta')})$ each of component is supported on $(\kappa, 1- \kappa)$. For example, in the binary case $\mathrm{T} = \{0, 1 \}$  we have $\beta'=1$ and $\beta=0$ and 
\begin{align*}
\phi (1, \nu_0(X)) - \phi (0, \nu_0(X)) = s_0(0,X)  - s_0(1,X)  p_Y(X).
\end{align*}
Assumption \ref{ass:fsrate2} reduces to assuming that $(s_0(0,X), s_0(1,X), p_Y(X))$ has a bounded joint density.  Given a vector $(X,Y,Z)$ with a bounded joint density supported on $(\kappa, 1-\kappa)^3$, the map $X - YZ$ has partial derivatives bounded away from zero since $\kappa>0$ and from above. Thus, $s_0(0,X)  - s_0(1,X)  p_Y(X)$ also has a bounded density.

\end{proof}

 \bibliographystyle{apalike}
\bibliography{/Users/virasemenova/Desktop/my_new_bibtex}

\begin{thebibliography}{}

\bibitem[Abramovich et~al.,
  2020]{abramovich2020multiclassclassificationsparsemultinomial}
Abramovich, F., Grinshtein, V., and Levy, T. (2020).
\newblock Multiclass classification by sparse multinomial logistic regression.

\bibitem[Acerenza et~al., 2023]{acerenza2023marginal}
Acerenza, S., Ban, K., and Kédagni, D. (2023).
\newblock Marginal treatment effects with a misclassified treatment.

\bibitem[Adjaho and Christensen, 2022]{adjaho2023externally}
Adjaho, C. and Christensen, T. (2022).
\newblock Externally valid policy choice.

\bibitem[Andrews and Shi, 2013]{AndrewsShiECMA}
Andrews, D. and Shi, X. (2013).
\newblock Inference based on conditional moment inequalities.
\newblock {\em Econometrica}, 81:609--666.

\bibitem[Athey and Wager, 2021]{AtheyWager2}
Athey, S. and Wager, S. (2021).
\newblock Policy learning with observational data.
\newblock {\em Econometrica}, 89:133--161.

\bibitem[Babii et~al., 2021]{babii2021binary}
Babii, A., Chen, X., Ghysels, E., and Kumar, R. (2021).
\newblock Binary choice with asymmetric loss in a data-rich environment: Theory
  and an application to racial justice.

\bibitem[Balke and Pearl, 1994]{BalkePearl1994}
Balke, A. and Pearl, J. (1994).
\newblock Counterfactual probabilities: Computational methods, bounds and
  applications.
\newblock pages 46--54. Morgan Kaufmann Publishers Inc.

\bibitem[Balke and Pearl, 1997]{BalkePearl1997}
Balke, A. and Pearl, J. (1997).
\newblock Bounds on treatment effects from studies with imperfect compliance.
\newblock {\em Journal of the American Statistical Association},
  92(439):1171--1176.

\bibitem[Ban and Kédagni, 2021]{ban2021nonparametric}
Ban, K. and Kédagni, D. (2021).
\newblock Nonparametric bounds on treatment effects with imperfect instruments.

\bibitem[Bartalotti et~al., 2021]{bartalotti2021identifying}
Bartalotti, O., Kédagni, D., and Possebom, V. (2021).
\newblock Identifying marginal treatment effects in the presence of sample
  selection.

\bibitem[Belloni et~al., 2017]{Program}
Belloni, A., Chernozhukov, V., Fernandez-Val, I., and Hansen, C. (2017).
\newblock Program evaluation and causal inference with high-dimensional data.
\newblock {\em Econometrica}, 85:233--298.

\bibitem[Belloni et~al., 2016]{orthogStructural}
Belloni, A., Chernozhukov, V., and Wei, Y. (2016).
\newblock Post-selection inference for generalized linear models with many
  controls.
\newblock {\em Journal of Business \& Economic Statistics}, 34(4):606--619.

\bibitem[Ben-Michael et~al., 2024]{benmichael2023policy}
Ben-Michael, E., Imai, K., and Jiang, Z. (2024).
\newblock Policy learning with asymmetric counterfactual utilities.
\newblock {\em Journal of the American Statistical Association}, 0(0):1--14.

\bibitem[Beresteanu et~al., 2011]{BMM2}
Beresteanu, A., Molchanov, I., and Molinari, F. (2011).
\newblock Sharp identification regions in models with convex moment
  predictions.
\newblock {\em Econometrica}, 79:1785--1821.

\bibitem[Beresteanu and Molinari, 2008]{BM}
Beresteanu, A. and Molinari, F. (2008).
\newblock Asymptotic properties for a class of partially identified models.
\newblock {\em Econometrica}, 76(4):763--814.

\bibitem[Bloom et~al., 1997]{Bloom1997}
Bloom, H.~S., Orr, L.~L., Bell, S.~H., Cave, G., Doolittle, F., Lin, W., and
  Bos, J.~M. (1997).
\newblock The benefits and costs of jtpa title ii-a programs: Key findings from
  the national job training partnership act study.
\newblock {\em The Journal of Human Resources}, 32(3):549--576.

\bibitem[Bonvini et~al., 2022]{Bonvini_2022}
Bonvini, M., Kennedy, E., Ventura, V., and Wasserman, L. (2022).
\newblock Sensitivity analysis for marginal structural models.

\bibitem[Bonvini and Kennedy, 2021]{Bonvini_2021}
Bonvini, M. and Kennedy, E.~H. (2021).
\newblock Sensitivity analysis via the proportion of unmeasured confounding.
\newblock {\em Journal of the American Statistical Association}, page 1–11.

\bibitem[{Chandrasekhar} et~al., 2012]{CCMS}
{Chandrasekhar}, A., {Chernozhukov}, V., {Molinari}, F., and {Schrimpf}, P.
  (2012).
\newblock {Inference for best linear approximations to set identified
  functions}.
\newblock {\em arXiv e-prints}, page arXiv:1212.5627.

\bibitem[Chernozhukov et~al., 2018]{chernozhukov2016double}
Chernozhukov, V., Chetverikov, D., Demirer, M., Duflo, E., Hansen, C., Newey,
  W., and Robins, J. (2018).
\newblock Double/debiased machine learning for treatment and structural
  parameters.
\newblock {\em Econometrics Journal}, 21:C1--C68.

\bibitem[{Chernozhukov} et~al., 2022]{LRSP}
{Chernozhukov}, V., {Escanciano}, J.~C., {Ichimura}, H., {Newey}, W.~K., and
  {Robins}, J.~M. (2022).
\newblock {Locally Robust Semiparametric Estimation}.
\newblock {\em Econometrica}.

\bibitem[Chernozhukov et~al., 2007]{CHT}
Chernozhukov, V., Hong, H., and Tamer, E. (2007).
\newblock Estimation and confidence regions for parameter sets in econometric
  models.
\newblock {\em Econometrica}, 75:1243--1284.

\bibitem[Chernozhukov et~al., 2013]{CLR}
Chernozhukov, V., Lee, S., and Rosen, A. (2013).
\newblock {\em Econometrica}, 81(2):667–737.

\bibitem[Chernozhukov et~al., 2015]{CherNeweySantos}
Chernozhukov, V., Newey, W.~K., and Santos, A. (2015).
\newblock Constrained conditional moment restriction models.

\bibitem[Christensen et~al., 2023]{christensen2023optimal}
Christensen, T., Moon, H.~R., and Schorfheide, F. (2023).
\newblock Optimal decision rules when payoffs are partially identified.

\bibitem[Cilibero and Tamer, 2009]{CilibertoTamer}
Cilibero, F. and Tamer, E. (2009).
\newblock Market structure and multiple equilibria in airline markets.
\newblock {\em Econometrica}, 77:1791--1828.

\bibitem[{Colangelo} and {Lee}, 2020]{Colangelo}
{Colangelo}, K. and {Lee}, Y.-Y. (2020).
\newblock {Double Debiased Machine Learning Nonparametric Inference with
  Continuous Treatments}.
\newblock {\em arXiv e-prints}, page arXiv:2004.03036.

\bibitem[Cui, 2021]{Cui_2021}
Cui, Y. (2021).
\newblock Individualized decision making under partial identification:
  Threeperspectives, two optimality results, and one paradox.
\newblock {\em Harvard Data Science Review}.

\bibitem[Cui and Han, 2024]{cui2024policy}
Cui, Y. and Han, S. (2024).
\newblock Policy learning with distributional welfare.

\bibitem[D'Adamo, 2022]{dadamo2022orthogonal}
D'Adamo, R. (2022).
\newblock Orthogonal policy learning under ambiguity.

\bibitem[Dorn and Guo, 2021]{DornGuo}
Dorn, J. and Guo, K. (2021).
\newblock Sharp sensitivity analysis for inverse propensity weighting via
  quantile balancing.

\bibitem[Dorn et~al., 2021]{DornGuoKallus}
Dorn, J., Guo, K., and Kallus, N. (2021).
\newblock Doubly-valid/doubly-sharp sensitivity analysis for causal inference
  with unmeasured confounding.

\bibitem[Fan et~al., 2017]{FanZhu}
Fan, Y., Guerre, E., and Zhu, D. (2017).
\newblock Partial identification of functionals of the joint distribution of
  ``potential outcomes''.

\bibitem[Fan and Park, 2010]{FanPark}
Fan, Y. and Park, S.~S. (2010).
\newblock Sharp bounds on the distribution of treatment effects and their
  statistical inference.
\newblock {\em Econometric Theory}, 26(3):931--951.

\bibitem[Fan and Park, 2012]{FanPark2}
Fan, Y. and Park, S.~S. (2012).
\newblock Confidence intervals for the quantile of treatment effects in
  randomized experiments.
\newblock {\em Journal of Econometrics}, 167:330--344.

\bibitem[Fang and Santos, 2018]{FangSantos}
Fang, Z. and Santos, A. (2018).
\newblock {Inference on Directionally Differentiable Functions}.
\newblock {\em The Review of Economic Studies}, 86(1):377--412.

\bibitem[Fava, 2024]{fava2024}
Fava, B. (2024).
\newblock Predicting the distribution of treatment effects: A
  covariate-adjustment approach.

\bibitem[Finkelstein et~al., 2012]{finkelstein}
Finkelstein, A., Taubman, S., Wright, B., Bernstein, M., Gruber, J., Newhouse,
  J., Allen, H., Baicker, K., and Group, O. H.~S. (2012).
\newblock The oregon health insurance experiment: Evidence from the first year.
\newblock {\em Quarterly Journal of Economics}, 127(3):1057--1106.

\bibitem[Firpo et~al., 2023]{firpo2023loss}
Firpo, S., Galvao, A.~F., Kobus, M., Parker, T., and Rosa-Dias, P. (2023).
\newblock Loss aversion and the welfare ranking of policy interventions.

\bibitem[Firpo and Ridder, 2019]{FirpoRidder}
Firpo, S. and Ridder, G. (2019).
\newblock Partial identification of the treatment effect distribution and its
  functionals.
\newblock {\em Journal of Econometrics}, 213:210--234.

\bibitem[Gafarov, 2019]{Gafarov}
Gafarov, B. (2019).
\newblock Inference in high-dimensional set-identified affine models.

\bibitem[Haile and Tamer, 2003]{HaileTamer}
Haile, P.~A. and Tamer, E. (2003).
\newblock Inference with an incomplete model of english auctions.
\newblock {\em Journal of Political Economy}, 111(1):1--51.

\bibitem[Han, 2021]{han2021optimal}
Han, S. (2021).
\newblock Optimal dynamic treatment regimes and partial welfare ordering.

\bibitem[Heckman et~al., 1997]{HSC}
Heckman, J., Smith, J., and Clements, N. (1997).
\newblock Making the most out of program evaluations and social experiments:
  accounting for heterogeneity in program impacts.
\newblock {\em Review of Economic Studies}, 64:487--535.

\bibitem[Henry et~al., 2023]{henry2023role}
Henry, M., Meango, R., and Mourifie, I. (2023).
\newblock Role models and revealed gender-specific costs of stem in an extended
  roy model of major choice.

\bibitem[Hirano et~al., 2003]{HIR2003}
Hirano, K., Imbens, G., and Reeder, G. (2003).
\newblock Efficient estimation of average treatment effects under the estimated
  propensity score.
\newblock {\em Econometrica}, 71(4):1161--1189.

\bibitem[Hirano and Porter, 2012]{HiranoPorter2012}
Hirano, K. and Porter, J. (2012).
\newblock Impossibility results for nondifferentiable functionals.
\newblock {\em Econometrica}, 80:1769--1790.

\bibitem[Horowitz and Manski, 1995]{HorowitzManski}
Horowitz, J.~L. and Manski, C.~F. (1995).
\newblock Identification and robustness with contaminated and corrupted data.
\newblock {\em Econometrica}, 63(2):281--302.

\bibitem[Ishihara and Kitagawa, 2021]{ishihara2021evidence}
Ishihara, T. and Kitagawa, T. (2021).
\newblock Evidence aggregation for treatment choice.

\bibitem[{Jeong} and {Namkoong}, 2020]{jeong2020robust}
{Jeong}, S. and {Namkoong}, H. (2020).
\newblock {Robust causal inference under covariate shift via worst-case
  subpopulation treatment effects}.
\newblock {\em arXiv e-prints}, page arXiv:2007.02411.

\bibitem[Ji et~al., 2023]{JLS}
Ji, W., Lei, L., and Spector, A. (2023).
\newblock Model-agnostic covariate-assisted inference on partially identified
  causal effects.

\bibitem[Kallus, 2022]{kallus2022whats}
Kallus, N. (2022).
\newblock What's the harm? sharp bounds on the fraction negatively affected by
  treatment.

\bibitem[Kallus et~al., 2020a]{kallus2020localized}
Kallus, N., Mao, X., and Uehara, M. (2020a).
\newblock Localized debiased machine learning: Efficient inference on quantile
  treatment effects and beyond.

\bibitem[Kallus et~al., 2020b]{kallus2020assessing}
Kallus, N., Mao, X., and Zhou, A. (2020b).
\newblock Assessing algorithmic fairness with unobserved protected class using
  data combination.

\bibitem[Kallus and Zhou, 2019]{kallus2019assessing}
Kallus, N. and Zhou, A. (2019).
\newblock Assessing disparate impacts of personalized interventions:
  Identifiability and bounds.

\bibitem[Kennedy, 2022]{Kennedy_review}
Kennedy, E.~H. (2022).
\newblock Semiparametric doubly robust targeted double machine learning: a
  review.

\bibitem[Kitagawa et~al., 2023]{kitagawa2023treatment}
Kitagawa, T., Lee, S., and Qiu, C. (2023).
\newblock Treatment choice, mean square regret and partial identification.

\bibitem[Kitagawa and Tetenov, 2018]{KitagawaTetenov}
Kitagawa, T. and Tetenov, A. (2018).
\newblock Who should be treated? empirical welfare maximization methods for
  treatment choice.
\newblock {\em Econometrica}, 86:591--616.

\bibitem[Kolesar, 2013]{Kolesar}
Kolesar, M. (2013).
\newblock Estimation in an instrumental variable model with treatment effect
  heterogeneity.

\bibitem[Kroft et~al., 2024]{kroft2024leeboundsmultilayeredsample}
Kroft, K., Mourifié, I., and Vayalinkal, A. (2024).
\newblock Lee bounds with multilayered sample selection.

\bibitem[Lee, 2009]{LeeBound}
Lee, D. (2009).
\newblock Training, wages, and sample selection: Estimating sharp bounds on
  treatment effects.
\newblock {\em Review of Economic Studies}, 76(3):1071--1102.

\bibitem[Lee, 2021]{LeeSungwon}
Lee, S. (2021).
\newblock Partial identification and inference for conditional distributions of
  treatment effects.

\bibitem[Levis et~al., 2023]{levis2023covariateassisted}
Levis, A.~W., Bonvini, M., Zeng, Z., Keele, L., and Kennedy, E.~H. (2023).
\newblock Covariate-assisted bounds on causal effects with instrumental
  variables.

\bibitem[Li et~al., 2022]{li2022discordant}
Li, L., Kédagni, D., and Mourifié, I. (2022).
\newblock Discordant relaxations of misspecified models.

\bibitem[Luedtke and van~der Laan, 2016]{LuedtkeLaan}
Luedtke, A. and van~der Laan, M. (2016).
\newblock Statistical inference for the mean outcome under a possibly
  non-unique optimal treatment strategy.
\newblock {\em Annals of Statistics}, 44(2):713--742.

\bibitem[Makarov, 1981]{Makarov}
Makarov, G. (1981).
\newblock Estimates for the distribution function of a sum of two random
  variables when the marginal distributions are fixed.
\newblock {\em Theory of Probability and its Applications}, 26:803--806.

\bibitem[Mammen and Tsybakov, 1999]{MammenTsybakov}
Mammen, E. and Tsybakov, A.~B. (1999).
\newblock {Smooth discrimination analysis}.
\newblock {\em The Annals of Statistics}, 27(6):1808 -- 1829.

\bibitem[Manski, 1997]{Manski}
Manski, C. (1997).
\newblock Monotone treatment response.
\newblock {\em Econometrica}, 65(6):1311--1334.

\bibitem[Manski and Pepper, 2000]{ManskiPepper}
Manski, C. and Pepper, J. (2000).
\newblock Monotone instrumental variables: With an application to the returns
  to schooling.
\newblock {\em Econometrica}, 68(4):997--1010.

\bibitem[Manski and Tamer, 2002]{Manski:2002}
Manski, C. and Tamer, E. (2002).
\newblock Inference on regressions with interval data on a regressor or
  outcome.
\newblock {\em Econometrica}, 70(2):519--546.

\bibitem[Manski, 1989]{Manski89}
Manski, C.~F. (1989).
\newblock Anatomy of the selection problem.
\newblock {\em The Journal of Human Resources}, 24(3):343--360.

\bibitem[Manski, 1990]{Manski90}
Manski, C.~F. (1990).
\newblock Nonparametric bounds on treatment effects.
\newblock {\em The American Economic Review}, 80(2):319--323.

\bibitem[Manski, 2004]{Manski2004}
Manski, C.~F. (2004).
\newblock Statistical treatment rules for heterogeneous populations.
\newblock {\em Econometrica}, 72:1221--1246.

\bibitem[Mbakop and Tabord-Meehan, 2021]{MbakopTabord}
Mbakop, E. and Tabord-Meehan, M. (2021).
\newblock Model selection for treatment choice: Penalized welfare maximization.
\newblock {\em Econometrica}, 89:825--848.

\bibitem[Molchanov and Molinari, 2018]{Molinari:2018}
Molchanov, I. and Molinari, F. (2018).
\newblock {\em Random Sets in Econometrics}.
\newblock Cambridge University Press.

\bibitem[Molinari, 2008]{Molinari2008}
Molinari, F. (2008).
\newblock Partial identification of probability distributions with
  misclassified data.
\newblock {\em Journal of Econometrics}, 144:81--117.

\bibitem[Molinari, 2020]{MolinariHandbook}
Molinari, F. (2020).
\newblock Chapter 5 - microeconometrics with partial identification.
\newblock In {\em Handbook of Econometrics, Volume 7A}, volume~7 of {\em
  Handbook of Econometrics}, pages 355--486. Elsevier.

\bibitem[Mourifi\'e et~al., 2020]{MourifieHenry}
Mourifi\'e, I., Henry, M., and Meango, R. (2020).
\newblock Sharp bounds and testability of a roy model of stem major choices.
\newblock {\em Journal of Political Economy}, 3(128):3220--3283.

\bibitem[Newey, 1994]{Newey1994}
Newey, W. (1994).
\newblock The asymptotic variance of semiparametric estimators.
\newblock {\em Econometrica}, 62(6):245--271.

\bibitem[Neyman, 1959]{Neyman:1959}
Neyman, J. (1959).
\newblock Optimal asymptotic tests of composite statistical hypotheses.
\newblock {\em Probability and Statistics}, 213(57):416--444.

\bibitem[Neyman, 1979]{Neyman:1979}
Neyman, J. (1979).
\newblock $c(\alpha)$ tests and their use.
\newblock {\em Sankhya}, pages 1--21.

\bibitem[Olea et~al., 2023]{olea2023decision}
Olea, J. L.~M., Qiu, C., and Stoye, J. (2023).
\newblock Decision theory for treatment choice problems with partial
  identification.

\bibitem[Ponomarev and Semenova, 2024]{ponomarev2024}
Ponomarev, K. and Semenova, V. (2024).
\newblock On the lower confidence band for the optimal welfare.

\bibitem[Powell, 1984]{Powell}
Powell, J.~L. (1984).
\newblock Least absolute deviations estimation for the censored regression
  model.
\newblock {\em Journal of Econometrics}, 25:303--325.

\bibitem[Pu and Zhang, 2021]{Pu_2021}
Pu, H. and Zhang, B. (2021).
\newblock Estimating optimal treatment rules with an instrumental variable: A
  partial identification learning approach.
\newblock {\em Journal of the Royal Statistical Society Series B: Statistical
  Methodology}, 83(2):318–345.

\bibitem[Qian and Murphy, 2011]{QianMurphy}
Qian, M. and Murphy, S.~A. (2011).
\newblock {Performance guarantees for individualized treatment rules}.
\newblock {\em The Annals of Statistics}, 39(2):1180 -- 1210.

\bibitem[Robins and Rotnitzky, 1995]{Robins}
Robins, J. and Rotnitzky, A. (1995).
\newblock Semiparametric efficiency in multivariate regression models with
  missing data.
\newblock {\em Journal of American Statistical Association}, 90(429):122--129.

\bibitem[Rockafellar and Uryasev, 2000]{RockUryasev}
Rockafellar, R.~T. and Uryasev, S. (2000).
\newblock Optimization of conditional value-at-risk.
\newblock {\em Journal of Risk}, 2(3):21--41.

\bibitem[Schick, 1986]{schick1986asymptotically}
Schick, A. (1986).
\newblock On asymptotically efficient estimation in semiparametric models.
\newblock {\em The Annals of Statistics}, 14(3):1139--1151.

\bibitem[Semenova, 2020]{SemSupp2}
Semenova, V. (2020).
\newblock Generalized lee bounds.

\bibitem[Semenova, 2023]{SemJoE}
Semenova, V. (2023).
\newblock Debiased machine learning for set-identified linear models.
\newblock {\em Journal of Econometrics}.

\bibitem[Semenova and Chernozhukov, 2021]{SemCher}
Semenova, V. and Chernozhukov, V. (2021).
\newblock Debiased machine learning of conditional average treatment effects
  and other causal functions.

\bibitem[Stoye, 2009]{Stoye}
Stoye, J. (2009).
\newblock Minimax regret treatment choice with finite samples.
\newblock {\em Journal of Econometrics}, 151:70--81.

\bibitem[Sun, 2021]{Sun}
Sun, L. (2021).
\newblock Empirical welfare maximization with constraints.

\bibitem[Tetenov, 2012]{Tetenov}
Tetenov, A. (2012).
\newblock Identification of positive treatment effects in randomized
  experiments with non-compliance.

\bibitem[Tsybakov, 2004]{Tsybakov}
Tsybakov, A.~B. (2004).
\newblock {Optimal aggregation of classifiers in statistical learning}.
\newblock {\em The Annals of Statistics}, 32(1):135 -- 166.

\bibitem[Yata, 2023]{yata2023optimal}
Yata, K. (2023).
\newblock Optimal decision rules under partial identification.

\bibitem[Zhang et~al., 2009]{ZhangRubin}
Zhang, J.~L., Rubin, D.~B., and Mealli, F. (2009).
\newblock Likelihood-based analysis of causal effects of job-training programs
  using principal stratification.
\newblock {\em Journal of American Statistical Association}, 104(85):166--176.

\end{thebibliography}

\end{document}